\documentclass{article}

\usepackage{fullpage}

\usepackage[utf8]{inputenc} %
\usepackage[style=alphabetic,natbib=true,maxcitenames=2, maxbibnames=10,backend=bibtex]{biblatex}
\usepackage[T1]{fontenc}    %
\usepackage[hypertexnames=false]{hyperref}
\usepackage{url}            %
\usepackage{booktabs}       %
\usepackage{amsfonts}       %
\usepackage{nicefrac}       %
\usepackage{microtype}      %
\usepackage{xcolor}         %
\usepackage{amsmath}
\usepackage{cleveref}
\usepackage{cases}
\usepackage{amsthm}
\usepackage{xspace}
\usepackage{graphicx}
\usepackage{comment}
\usepackage{enumitem}
\usepackage[final]{showlabels}
\usepackage{tikz}
\usepackage{nicefrac}
\usepackage{thm-restate}
\usepackage{pgfplots}
\usepackage{subcaption}
\usepackage{nicefrac}
\pgfplotsset{compat=1.18}
\usetikzlibrary{calc, positioning, arrows.meta}
\newcommand{\PURIFY}{\textup{\textsc{Purify}}\xspace}
\newcommand{\NOR}{\textup{\textsc{Nor}}\xspace}
\newcommand{\ORACLE}{\textup{\textsc{Oracle}}\xspace}
\newcommand{\CFont}[1]{{\textup{\textsc{#1}}}\xspace}
\newcommand{\OPC}{\CFont{OraclePureCircuit}}
\newcommand{\PC}{\CFont{PureCircuit}}
\newcommand{\SB}{\CFont{SmoothBrouwer}}
\newcommand{\SP}{\CFont{StrongSperner}}
\newcommand{\GDA}{\CFont{GDA-FixedPoint}}

\newcommand{\In}{\textnormal{In}}
\newcommand{\Out}{\textnormal{Out}}
\newcommand{\Reals}{\mathbb{R}}
\newcommand{\poly}{\textnormal{poly}}

\newcommand{\Ccal}{\mathcal{C}}
\newcommand{\cG}{\mathcal{G}}
\newcommand{\cL}{\mathcal{L}}
\newcommand{\cE}{\mathcal{E}}
\newcommand{\cH}{\mathcal{H}}

\newcommand{\wid}{M}
\newcommand{\dims}{d}
\newcommand{\cop}{K}

\newtheorem{theorem}{Theorem}[section]
\newtheorem{lemma}{Lemma}
\newtheorem{corollary}{Corollary}
\theoremstyle{definition}
\newtheorem{definition}{Definition}
\newtheorem{problem}{Problem}
\newtheorem{remark}{Remark}

\Crefname{problem}{Problem}{Problems}

\definecolor[named]{linkcolor}{cmyk}{0.55,1,0,0.15}
\definecolor[named]{citecolor}{cmyk}{1,0.58,0,0.21}
\hypersetup{
    colorlinks=true,
    linkcolor=linkcolor,
    filecolor=linkcolor,      
    citecolor=citecolor,
    urlcolor=black,
}

\title{
Min-Max Optimization Requires Exponentially Many Queries
}

\author{
\begin{tabular}{cc}
& \\
{Martino Bernasconi}\thanks{Martino Bernasconi and Andrea Celli were supported by an ERC grant (Project 101165466 — PLA-STEER).} & {Matteo Castiglioni}\thanks{Matteo Castiglioni was supported by the FAIR (Future Artificial Intelligence Research) project PE0000013, funded by the NextGenerationEU program within the PNRRPE-AI scheme (M4C2, Investment 1.3, Line on Artificial Intelligence), and by the EU Horizon project ELIAS (European Lighthouse of AI for Sustainability, No. 101120237).}\\
\small{Bocconi University} & \small{Politecnico di Milano}\\
{\textcolor{black}{\small\texttt{martino.bernasconi@unibocconi.it}}} & %
{\textcolor{black}{\small\texttt{matteo.castiglioni@polimi.it}}}\\
& \\
{Andrea Celli}\footnotemark[1] & {Alexandros Hollender}\\
\small{Bocconi University} & \small{University of Oxford}\\
{\textcolor{black}{\small\texttt{andrea.celli2@unibocconi.it}}} & %
{\textcolor{black}{\small\texttt{alexandros.hollender@cs.ox.ac.uk}}}\\
& \\
\end{tabular}
}

\date{}

\begin{document}

\maketitle

\begin{abstract}
\noindent We study the query complexity of min-max optimization of a nonconvex-nonconcave function $f$ over $[0,1]^d \times [0,1]^d$. We show that, given oracle access to $f$ and to its gradient $\nabla f$, any algorithm that finds an $\varepsilon$-approximate stationary point must make a number of queries that is exponential in $1/\varepsilon$ or~$d$.
\end{abstract}

\newpage
\section{Introduction}
In this paper, we consider the following constrained min-max optimization problem
\begin{equation}\label{eq:minmax-intro}
\min_{x \in X} \max_{y \in Y} \, f(x,y)
\end{equation}
where $X, Y \subset \Reals^d$ are compact convex sets, and $f: X \times Y \to \Reals$ is continuously differentiable with Lipschitz-continuous gradient $\nabla f$.\footnote{Here, and in the rest of this discussion, we assume that both $f$ and $\nabla f$ are $1$-Lipschitz continuous. This is without loss of generality, as it can be achieved by a simple rescaling of the function. Similarly, we also assume that $X, Y \subseteq [0,1]^d$, which can also be ensured by rescaling the domain. This will allow us to express bounds solely in terms of the approximation error and the dimension $d$.} %
This problem has wide practical relevance, as it is a fundamental component in optimization and machine learning applications such as generative adversarial networks \citep{goodfellow2014generative}, robustness to adversarial attacks \citep{madry2018towards, razaviyayn2020nonconvex}, and, most recently, large language model \citep{swamy2024minimaximalist, liu2024comal, munos2024nash, sun2025game, paulus2025safety} (see also \citet{daskalakis2022non} for a broader discussion on the importance of this problem in optimization and machine learning).

The problem can be solved efficiently in the case where $f$ is %
\emph{convex-concave}, i.e., when $x \mapsto f(x,y)$ is convex for all $y \in Y$, and $y \mapsto f(x,y)$ is concave for all $x \in X$. Namely, an $\varepsilon$-approximate (global) optimal solution of \eqref{eq:minmax-intro} can be identified using at most $\poly(d, \log(1/\varepsilon))$ queries to $f$ and $\nabla f$ \citep{rodomanov2023subgradient, anagnostides2025polynomial}. %
However, as soon as we let the function $f$ be only \emph{nonconvex-concave} (or, analogously, \emph{convex-nonconcave}), the problem already becomes harder to solve. Indeed, it is easy to see that it is at least as hard as a constrained nonconvex minimization problem of the form
\begin{equation}\label{eq:min-intro}
\min_{x \in X} \, g(x)
\end{equation}
where $g$ is continuously differentiable with Lipschitz-continuous gradient $\nabla g$. This is immediate by setting $f(x,y) := g(x)$. Importantly, since $g$ can be nonconvex, it is known that we cannot hope to efficiently locate the approximate position of a global minimum, or even a point with an objective function value close to that of a global minimum \citep{vavasis1995complexity}.

Instead, a more tractable solution concept for the optimization problem \eqref{eq:min-intro} is an approximate \emph{local} minimum, or, more precisely, an approximate constrained first-order stationary point. This is also called an approximate Karush-Kuhn-Tucker (KKT) point, and corresponds to an $\varepsilon$-approximate fixed point of the projected gradient descent operator. It is well known that such a point can be found with at most $O(1/\varepsilon^2)$ first-order queries, for example, by projected gradient descent. Moreover, this dependence on $\varepsilon$ is unavoidable in the worst case: no algorithm can find an $\varepsilon$-approximate KKT point of \eqref{eq:min-intro} with fewer than $\poly(1/\varepsilon)$ queries \citep{vavasis1993black,carmon2020lower}.

A similar picture holds for the nonconvex-concave case of the min-max problem \eqref{eq:minmax-intro}. In this setting, one can still compute an appropriate local solution using at most $\poly(d,1/\varepsilon)$ first-order queries \citep{nouiehed2019solving, lin2020gradient, ostrovskii2021efficient}. Thus, in both nonconvex minimization and nonconvex-concave min-max optimization, local solution concepts remain algorithmically tractable, although the dependence on the accuracy is necessarily polynomial.

Unfortunately, in most practical applications the function $f(x,y)$ is neither convex in $x$, nor concave in $y$. Also in this setting, the focus shifts toward the less ambitious goal of finding local solutions to problem \eqref{eq:minmax-intro}. Many notions of local solutions have been proposed for the problem \eqref{eq:minmax-intro}, see, e.g., \citet*{jin2020local} for an overview. In this work, we focus on $\varepsilon$-approximate (constrained) stationary points, namely a tuple $(x,y)\in X\times Y$ such that 
\[-\nabla f_x(x,y)^\top(x'-x)\le\varepsilon \quad \text{and}\quad \nabla f_y(x,y)^\top(y'-y)\le\varepsilon \quad \forall (x',y')\in X\times Y.\] This is a weak first-order condition: it merely requires that neither player has a feasible first-order improving direction larger than $\varepsilon$. This notion is very weak in the sense that it is necessary for the standard local optimality notions considered in the min-max literature, including local min-max points and local saddle points.\footnote{Unfortunately, the terminology for these solution concepts is not fully standardized across the literature.} This is favorable in our case, since we will provide lower bounds, thereby strengthening our results. 
First-order stationary points can also be viewed as fixed points of the projected gradient descent-ascent map, just as KKT points of a minimization problem are fixed points of projected gradient descent. However, in contrast to minimization problems, it is well known that gradient descent-ascent dynamics exhibit cycling behavior and thus do not yield an algorithm guaranteed to find such solutions \citep{mertikopoulos2018cycles, daskalakis2018limit}.

This contrast naturally leads to the following question:
\begin{quote}
When $f$ is nonconvex-nonconcave, can an $\varepsilon$-approximate first-order stationary point of the min-max problem \eqref{eq:minmax-intro} be found using at most $\poly(d, 1/\varepsilon)$ queries?
\end{quote}
In this paper, we answer this question in the negative. Namely, we show that any algorithm finding an $\varepsilon$-approximate solution must make a number of queries that is exponential in $d$ or in $1/\varepsilon$. This shows that finding approximate stationary points of \eqref{eq:minmax-intro} is significantly harder than in the corresponding minimization problem \eqref{eq:min-intro}. In particular, our result implies that nonconvex-nonconcave min-max optimization does not admit an efficient gradient-descent-type algorithm.

A negative result of this form was already known for the more general case of \emph{joint} constraints, i.e.,
\begin{equation}\label{eq:minmax-joint-intro}
\begin{split}
&\min_{x \in \Reals^d} \max_{y \in \Reals^d} \, f(x,y) \\
&\text{s.t.} \quad (x,y) \in P
\end{split}
\end{equation}
where $P$ is a convex polytope. \citet*{daskalakis2021complexity} proved that any algorithm that finds an $\varepsilon$-approximate solution of \eqref{eq:minmax-joint-intro} must make a number of queries that is exponential in $d$ or in $1/\varepsilon$. Note that problem \eqref{eq:minmax-joint-intro} is more general than \eqref{eq:minmax-intro}, because it allows \emph{joint} constraints between $x$ and $y$, meaning that $P$ cannot be decomposed as a product $P = X \times Y$. As discussed by \citet*{bernasconi2024roleconstraintscomplexityminmax}, the joint constraints are used in a crucial way in the work of \citet{daskalakis2021complexity}.\footnote{In fact, the lower bound proved by \citet{daskalakis2021complexity} for joint constraints even applies to the setting where $f$ is nonconvex-concave. Since such problems can be solved using $\poly(d,1/\varepsilon)$ queries under product constraints, this shows that the hardness result in \citep{daskalakis2021complexity} fundamentally relies on joint constraints.}
Therefore, their result does not apply to the more natural setting of problem \eqref{eq:minmax-intro}, i.e., when we have \emph{product} constraints (meaning that the feasible choices of $x$ and $y$ are independent at the level of the constraint set). Determining the query complexity under product constraints (i.e., our main open question stated above) remains a major open problem, as recently discussed in a column in the SIAM Activity Group on Optimization newsletter \citep{diakonikolas2025-SIAG-OPT}.

A very recent result of \citet{bernasconi2026complexity} established that finding approximate first-order stationary points of problem \eqref{eq:minmax-intro}, i.e., nonconvex-nonconcave min-max optimization with product constraints, is PPAD-hard. This means that when the function $f$ and its gradient $\nabla f$ are provided as arithmetic circuits or Turing machines, then the problem cannot be solved in polynomial time, unless PPAD $=$ P.
Importantly, as mentioned in \citet[Section~8]{bernasconi2026complexity}, their result does not imply a query lower bound for the problem, because their reduction is not \emph{black-box}: it reduces from a purely white-box problem called \PC \citep{deligkas2022pure}. They ask the question of whether it is possible to improve their reduction so as to also obtain a query lower bound for the problem. In this work, we show that this is indeed possible. Our improved reduction resolves the main open question stated above.

\paragraph{Our Contribution.}
We consider the min-max optimization problem \eqref{eq:minmax-intro} with the simple domain $X = Y = [0,1]^d$ and show that any algorithm that finds an $\varepsilon$-approximate first-order stationary point must make a number of queries to $f$ and $\nabla f$ that is exponential in $1/\varepsilon$ or $d$.  To be more precise, we prove the following theorem.\footnote{We use the term ``exponential in $x$'' in the weak sense, i.e., to mean $2^{\Omega({x^c})}$ for some absolute $c > 0$, as opposed to $2^{\Omega(x)}$.}

\begin{theorem}[Informal version of \Cref{thm:final}]\label{thm:main-intro}
Any algorithm that outputs an $\varepsilon$-approximate stationary point for problem \eqref{eq:minmax-intro}, must make at least a number of queries to $f$ or $\nabla f$ that is exponential in $d$, even when $\varepsilon$ is inversely polynomial in $d$.
\end{theorem}

This provides an \emph{unconditional} lower bound for the problem, as opposed to the conditional time-complexity lower bound of \citet{bernasconi2026complexity}, which only holds if PPAD $\neq$ P. Furthermore, our lower bound even rules out algorithms that are allowed to perform unbounded computation between queries. \citet*{chen2024computing} recently gave such a query-efficient but time-inefficient algorithm for the problem of computing a fixed point of a contraction map. \Cref{thm:main-intro} shows that no such algorithm exists for nonconvex-nonconcave min-max optimization.

\paragraph{Our Techniques.}
As mentioned above, the main obstacle to obtaining a query lower bound from the construction of \citet{bernasconi2026complexity} is the usage of the \PC problem in their reduction. The \PC problem was introduced by \citet*{deligkas2022pure} as a tool for proving very strong inapproximability results in the context of PPAD. Importantly, it is a purely \emph{white-box} problem, meaning there is no oracle in the problem definition and the algorithm has full knowledge of the instance. Thus, it cannot be used to prove a query lower bound for our min-max problem.

The natural thing to do is to try to replace the \PC problem in the reduction by some other Brouwer-like problem. The issue with this is that the reduction of \citet{bernasconi2026complexity} heavily relies on the fact that the gates of the \PC problem are very flexible and allow for a lot of error without breaking. This fact is also precisely why \PC has been so useful in proving PPAD-hardness for various approximation problems in game theory and beyond. Unfortunately, black-box Brouwer-like problems, for which query lower bounds are known, do not offer this flexibility.

We resolve this issue by introducing a black-box version of the \PC problem, which we call \OPC. The definition of this new problem is identical to \PC, except that the problem has been augmented with one additional type of gate, which we call an oracle gate. The \OPC problem allows us to achieve our two desiderata: (i) keeping the simplicity and flexibility of \PC, while also (ii) having a problem that has an exponential query lower bound. Indeed, the query lower bound for \OPC follows relatively easily by adapting the existing PPAD-hardness proof for \PC from \citep{deligkas2022pure}, as we show in \Cref{sec:OPC}.

What remains then is modifying the reduction of \citet{bernasconi2026complexity} so that it uses \OPC instead of \PC. In particular, this entails implementing the newly introduced oracle gate. Another part of the reduction of \citet{bernasconi2026complexity} breaks when we try to use it to prove query lower bounds. Indeed, they reduce from two different problems: one of which is the \PC discussed above, and the other one is a problem related to linear variational inequalities, which is also not a query-hard problem; however, the modification of this part of the reduction is more standard, and we solve it by defining a smooth version of Brouwer's fixed point theorem, as detailed in \Cref{sec:smooth,sec:SB}. Finally, the proof culminates in \Cref{sec:GDA}, where we put everything together to prove \Cref{thm:main-intro}.

\section{Preliminaries}\label{sec:prelims}

\paragraph{Mathematical Notation.}
For a function $F:[0,1]^d\to\Reals^m$, we say that $F$ is $G$-Lipschitz if $\|F(x)-F(x')\|\le G\|x-x'\|$, and $L$-smooth if
$\|J_F(x)-J_F(x')\|_2\le L\|x-x'\|_2$, where $J_F(x)\in\Reals^{m\times d}$ is the Jacobian and $\|\cdot\|$ denotes the Euclidean norm for vectors and the spectral norm for matrices.
For a matrix $A\in\Reals^{m\times n}$, let $\|A\|_{\max}:=\max_{i,j}|A_{ij}|$. For $x\in\Reals^d$ and $r\ge 0$, let $B_r^\infty(x):=\{y\in\Reals^d:\|x-y\|_\infty\le r\}$. For $f:\Reals^n\to\Reals$, we denote by $D^2f(x)\in\Reals^{n\times n}$ its Hessian. For $g:\Reals\to\Reals$, let $\|g\|_\infty:=\sup_{x\in\Reals}|g(x)|$. Finally, $[n]:=\{1,\dots,n\}$, and $\mathcal C^k$ denotes the class of functions whose partial derivatives up to order $k$ are continuous.

The main result of the paper is a query lower bound for the following problem, which concerns finding approximate fixed points of gradient descent-ascent dynamics (GDA).
\begin{problem}
[\GDA]\label{def:gdafp}
Given $\varepsilon$, $L$, $G, B$ $\in\Reals_+$, two oracles implementing a $G$-Lipschitz and $L$-smooth function $f:[0,1]^d\times[0,1]^d\to [-B,B]$ and its gradient $\nabla f: [0,1]^d\times[0,1]^d\to\Reals^{2d}$, find $(x^\star, y^\star)\in[0,1]^{d}\times[0,1]^d$ such that for all $i \in [d], x_i \in [0,1]$ and $y_i\in[0,1]$
\[ 
   -{\partial_{x_i} f(x^\star,y^\star)} ( x_i-x_i^\star)  \le \varepsilon \quad \text{and}\quad
   {\partial_{y_i} f(x^\star,y^\star)} ( y_i-y_i^\star) \le \varepsilon.
\]
\end{problem}

As also mentioned in the introduction, the naming of the problem is motivated by the problem being computationally equivalent to finding an approximate fixed-point of the GDA map $(x,y)\mapsto(\Pi_{[0,1]^d}(x-\nabla_xf(x,y)),\Pi_{[0,1]^d}(y+\nabla_yf(x,y)))$, as shown by \citet{daskalakis2021complexity}.

Our query lower bound ultimately comes from known lower bounds for finding a Brouwer fixed point. Specifically, we will use the following discrete Brouwer-like problem, called \SP.\footnote{This problem was first defined by \citet{daskalakis2021complexity}, where it was called \CFont{HighD-BiSperner}. Later, it was used by \citet{deligkas2022pure} under the name \SP to prove the PPAD-hardness of \PC.}

\begin{problem}[\SP]\label{def:SP}
Given integers $\wid, \dims$, and oracle access to a labeling $\lambda: [\wid]^\dims \to \{-1,+1\}^\dims$ satisfying the following boundary conditions for every $i \in [\dims]$:
\begin{align*}
    \text{$x_i = 1\implies[\lambda(x)]_i = +1$}\quad\text{and}\quad
    \text{$x_i = \wid\implies[\lambda(x)]_i = -1$},
\end{align*}
output $x^{(1)}, \ldots, x^{(\dims)} \in [\wid]^\dims$ such that $\max_{i,j\in [\dims]}\|x^{(i)} - x^{(j)}\|_\infty \leq 1$, and such that %
all labels are covered, i.e., for all $i \in [\dims]$ and $\ell \in \{-1,+1\}$ there exists $j \in [\dims]$ with $[\lambda(x^{(j)})]_i = \ell$.
\end{problem}

The following query lower bound for \SP easily follows from the seminal query lower bound of \citet*{hirsch1989exponential} for finding Brouwer fixed points.

\begin{restatable}{theorem}{SPlowerbound}\label{thm:SP-lower-bound}
There exists a sufficiently large constant integer $\wid$ such that any algorithm that outputs a solution to \SP with parameters $\wid$ and $\dims$ must make $2^{\Omega(\dims)}$ many queries to the labeling $\lambda$.
\end{restatable}

\citet[Corollary 9.2]{daskalakis2021complexity} prove a weaker version of \Cref{thm:SP-lower-bound} where $\wid$ is only assumed to be $O(\dims)$, instead of constant. However, it is not hard to see that the work of \citet{hirsch1989exponential} in fact implies that $\wid$ can be fixed to a constant. For completeness, we provide a proof of this in \cref{sec:app:prelims}.

\section{The \OPC Problem} \label{sec:OPC}

In this section, we introduce the \OPC problem, which is a black-box version of the \PC problem of \citet{deligkas2022pure}, and prove a query lower bound for it.

\begin{problem}[\OPC]
An instance of \OPC is given by a natural number $N$, a set of nodes $V$, three sets of gates $\cG_{\NOR}$, $\cG_{\PURIFY}$, and $\cG_{\ORACLE}$, and oracle access to a function $\cL:\{0,1\}^N\rightarrow \{0,1\}$ (with $N\le |V|$). Each gate is of the form $(u_1,u_2,\ldots)$ where $u_i$ are distinct nodes in $V$ with the following interpretation:
\begin{itemize}
    \item if $(u,v,w) \in \cG_{\NOR}$, then $u$ and $v$ are the inputs of the gate, and $w$ is its output.
    \item if $(u,v,w) \in \cG_{\PURIFY}$, then $u$ is the input of the gate, and $v$ and $w$ are its outputs.
    \item if $(u_1,\ldots,u_N,v) \in \cG_{\ORACLE}$, then $u_1,\ldots,u_N$ are the inputs of the gate and $v$ its output.
\end{itemize}

Each node is the output of exactly one gate. A solution to an instance of \OPC is an assignment $b: V\rightarrow \{0,1,\bot\}$ that satisfies all the gates, i.e., for each gate we have:
\begin{itemize}%
 \item if $(u, v, w) \in \cG_{\NOR}$, then $b$ satisfies:
    \begin{align*}
        &b(u) = b(v) = 0 \implies b(w) = 1\\
        &(b(u) = 1) \text{  or  } (b(v) = 1) \implies b(w) = 0,
    \end{align*}
    \item if $(u, v, w) \in \cG_{\PURIFY}$, then $b$ satisfies:
    \begin{align*}
        &\{b(v), b(w)\} \cap \{0,1\} \neq \emptyset\\
        &b(u) \in \{0,1\} \implies b(v) = b(w) = b(u),
    \end{align*}
    \item if $(u_1,\ldots,u_N,v) \in \cG_{\ORACLE}$, then $b$ satisfies:
    \begin{align*}
        (b(u_1), \dots, b(u_N)) \in \{0,1\}^N \implies b(v) = \cL(b(u_1), \dots, b(u_N)).
    \end{align*}
    
\end{itemize}
\end{problem}

The \PC problem originally defined by \citet{deligkas2022pure} is the same as our \OPC, but without the oracle gates $\cG_{\ORACLE}$. In \OPC we have introduced oracle gates in the least constraining way: a gate of type $\cG_{\ORACLE}$ is only required to be correct when all its inputs are actual bits (i.e., no $\bot$). Just like \PC, our version is also always guaranteed to have a solution. This can be proved directly, and it also follows from our reduction to \GDA, which is known to always have a solution \citep{daskalakis2021complexity}. As mentioned in the introduction, \PC is PPAD-complete, but does not admit a query lower bound, because it is a purely white-box problem. In contrast, for \OPC we can show the following.

\begin{restatable}{theorem}{OPCLB}\label{th:OPCLB}

Given an instance of \OPC with $|V|$ nodes, any algorithm that has black-box oracle access to $\cL$ and outputs a solution to \OPC requires  $2^{\Omega(|V|^{1/3})}$ queries to  $\cL$.
\end{restatable}

\begin{proof}[Proof sketch]
The query lower bound is obtained by constructing a reduction from \SP with dimension $\dims$, for which a $2^{\Omega(d)}$ query lower bound is known by \Cref{thm:SP-lower-bound}. To do this, we modify the existing reduction from the white-box version of \SP to \PC, provided by \citet{deligkas2022pure}. Namely, whenever the standard \PC gates are used to implement the Boolean circuit of the white-box version of \SP, we instead use the oracle gate to implement the oracle of the black-box version of \SP. Using the fact that the parameter $\wid$ of \SP can be assumed to be constant (by \Cref{thm:SP-lower-bound}), we argue that the number of gates in the \OPC instance is at most $O(\dims^3)$, which yields a query lower bound of $2^{\Omega(|V|^{1/3})}$. For the details of the proof, we refer to \Cref{sec:app:OPC}.
\end{proof}

\section{Smooth Interpolation of a Boolean-Valued Function} \label{sec:smooth}

\begin{figure*}[t!]
    \centering
    \begin{subfigure}[t]{0.48\textwidth}
        \centering
        \usetikzlibrary {arrows.meta}
\begin{tikzpicture}
\pgfplotsset{my style/.append style={line width=1.5pt}}
\begin{axis}[
            xtick={1/3,1/2,2/3},
            xticklabels={$c_1$,,$c_2$},
            ytick={0,1},
            yticklabels={$0$,$1$},
            xmin=0,
            xmax=1,
            ymin=-0.1,
            ymax=1.2,
            axis x line=middle, 
            axis y line=middle,
            xlabel={$x$}, 
            x label style={anchor=north},
            ylabel={$\phi_{c_1,c_2}(x)$},
            y label style={anchor=south west},
            width=1.2*144pt,
            height=1.2*89pt,
            axis line style={-Latex[round]}
            ]
\addplot[my style, domain=-1:1/3] {0};
\addplot[my style, domain=1/3:2/3] {e^(-1/(x-1/3))/(e^(-1/(x-1/3))+e^(-1/(2/3-x)))};
\addplot[my style, domain=2/3:2] {1};
\end{axis}
\end{tikzpicture}
        \caption{The smooth step function of \Cref{lem:basicinterp}.}
        \label{fig:smoothstep}
    \end{subfigure}%
    ~ 
    \begin{subfigure}[t]{0.5\textwidth}
        \centering
        \scalebox{0.36}{\begin{tikzpicture}[x=9cm,y=9cm]

\definecolor{lowc}{RGB}{59,130,246}
\definecolor{midc}{RGB}{229,231,235}
\definecolor{highc}{RGB}{239,68,68}

\pgfmathsetmacro{\os}{1/6}
\pgfmathsetmacro{\ot}{1/3}
\pgfmathsetmacro{\tt}{2/3}
\pgfmathsetmacro{\fs}{5/6}
\pgfmathsetmacro{\tmpone}{0.25}

\fill[midc] (0,0) rectangle (1,1);
\fill[lowc!35] (0,0) rectangle (\ot,\ot);
\fill[highc!35] (\tt,0) rectangle (1,\ot);
\fill[highc!35] (0,\tt) rectangle (\ot,1);
\fill[lowc!35] (\tt,\tt) rectangle (1,1);

\fill[lowc!70] (0,0) rectangle (\os,\os);
\fill[highc!70] (\fs,0) rectangle (1,\os);
\fill[highc!70] (0,\fs) rectangle (\os,1);
\fill[lowc!70] (\fs,\fs) rectangle (1,1);

\draw[thick] (0,0) rectangle (1,1);
\draw[dashed,gray] (\ot,0) -- (\ot,1);
\draw[dashed,gray] (\tt,0) -- (\tt,1);
\draw[dashed,gray] (0,\ot) -- (1,\ot);
\draw[dashed,gray] (0,\tt) -- (1,\tt);

\draw[dotted,gray] (\os,0) -- (\os,\os);
\draw[dotted,gray] (0,\os) -- (\os,\os);
\draw[dotted,gray] (\fs,0) -- (\fs,\os);
\draw[dotted,gray] (1,\os) -- (\fs,\os);
\draw[dotted,gray] (0,\fs) -- (\os,\fs);
\draw[dotted,gray] (\os,1) -- (\os,\fs);
\draw[dotted,gray] (\fs,1) -- (\fs,\fs);
\draw[dotted,gray] (1,\fs) -- (\fs,\fs);

\node[font=\huge] at (0.08,0.08) {$0$};
\node[font=\huge] at (0.92,0.08) {$1$};
\node[font=\huge] at (0.08,0.92) {$1$};
\node[font=\huge] at (0.92,0.92) {$0$};
\node[font=\Large,align=center] at (0.25,0.25) {$(0,\tfrac12)$};
\node[font=\Large,align=center] at (0.75,0.25) {$(\tfrac12,1)$};
\node[font=\Large,align=center] at (0.25,0.75) {$(\tfrac12,1)$};
\node[font=\Large,align=center] at (0.75,0.75) {$(0,\tfrac12)$};
\node[font=\huge] at (0.50,0.50) {$\tfrac12$};

\node[below,font=\huge] at (0.5,-0.07) {$x_1$};
\node[rotate=90,font=\huge] at (-0.07,0.5) {$x_2$};

\foreach \x/\lbl in {0/0,1/1}
  \draw (\x,0) -- +(0,-0.015) node[below,font=\huge] {$\lbl$};
\foreach \y/\lbl in {1/1}
  \draw (0,\y) -- +(-0.015,0) node[left,font=\huge] {$\lbl$};

\node[below,font=\huge,gray] at (\os,0) {$\tfrac16$};
\node[below,font=\huge,gray] at (\ot,0) {$\tfrac13$};
\node[below,font=\huge,gray] at (\tt,0) {$\tfrac23$};
\node[below,font=\huge,gray] at (\fs,0) {$\tfrac56$};

\node[left,font=\huge,gray] at (0,\os) {$\tfrac16$};
\node[left,font=\huge,gray] at (0,\ot) {$\tfrac13$};
\node[left,font=\huge,gray] at (0,\tt) {$\tfrac23$};
\node[left,font=\huge,gray] at (0,\fs) {$\tfrac56$};

\end{tikzpicture}}
        \caption{Schematic illustration of the interpolation presented in \Cref{lem:inteprolation} for $N=2$ and $\hat h(x_1,x_2)=\textnormal{XOR}(x_1,x_2)$.}
        \label{fig:interp}
    \end{subfigure}
    \caption{}
\end{figure*}

To bridge the gap between discrete Boolean circuits and continuous optimization, we require a mechanism for extending Boolean functions to the continuous domain. Specifically, we seek a $C^{\infty}$ interpolation over $[0,1]^N$ of a Boolean function defined only on $\{0,1\}^N$.
In pursuit of this gadget, we first consider a smooth-step function. Namely, given two values $c_1,c_2$, we use a standard $\Ccal^{\infty}$ non-analytic function $\phi_{c_1,c_2}$, which implements a smooth step-function that interpolates the step function that is $0$ below $c_1$ and $1$ above $c_2$. 

\begin{definition}[Smooth step function]\label{def:smoothstep}
    Let $c_2> c_1\ge 0$. The function $\phi_{c_1,c_2}:\Reals\to[0,1]$ is defined as
    \[
    \phi_{c_1,c_2}(x) = \frac{\eta(x-c_1)}{\eta(x-c_1)+\eta(c_2-x)},
    \]
    where $\eta:\Reals\to[0,1]$ is defined as $\eta(x)=\mathbb{I}(x>0)\exp(-1/x)$.
\end{definition}

The main properties of $\phi_{c_1,c_2}$ are summarized in the next lemma. Basically, it guarantees that $\phi_{c_1,c_2}$ implements a smooth step from $c_1$ to $c_2$ and that its derivative is bounded by a constant (depending on $c_1$ and $c_2$). An illustration of such a function can be found in \Cref{fig:smoothstep}.

\begin{restatable}{lemma}{basicinterp}\label{lem:basicinterp}
Given two $c_1<c_2$, the function $\phi_{c_1,c_2}$ is in $\Ccal^{\infty}$ and satisfies:
\[  \phi_{c_1,c_2} (x) \in \begin{cases}
\{0\} & \text{If } x \le c_1\\
(0,1) & \text{if } x \in (c_1,c_2)\\
\{1\} & \text{if } x \ge c_2
\end{cases}.\]
Moreover, $\|\phi'_{c_1,c_2}\|_\infty\le e^{\frac2{c_2-c_1}},$ and $\|\phi''_{c_1,c_2}\|_\infty\le 12e^{\frac4{c_2-c_1}}$.
\end{restatable}

The main use of the smooth step function in \Cref{def:smoothstep} is to extend a Boolean function $\hat h$ defined on the hypercube vertices $\{0,1\}^N$ to a continuous function $h$ over the domain $[0,1]^N$. Besides being consistent on the vertices of the hypercube, we require it to be robust: the function $h(x)=\hat h(v)$ whenever $x$ belongs to a ``large'' neighborhood of the vertex $v$. 
Moreover, another critical requirement for this interpolation is query efficiency: the value of the interpolated function at any point $x \in [0,1]^N$ must be computable using only a polynomial number of queries to the original function.

\begin{restatable}{lemma}{interpolation}\label{lem:inteprolation}
Let $\hat h: \{0,1\}^N \to \{0,1\}$ be a Boolean function. There exists a smooth $\mathcal{C}^\infty$ function $h= \cH(\cdot, \hat h) : [0,1]^N\rightarrow [0,1]$ that satisfies
    \begin{enumerate}[label=(\roman*)]
        \item $\hat h(y)=h(x)$ for all $y\in\{0,1\}^N$ and all $x\in B^\infty_{\nicefrac16}(y)\cap[0,1]^N$ ;
        \item $h(x)\in[0,1]$ for all $x\in[0,1]^N$;
        \item $h(x)$ and $\nabla h(x)$ are computable with at most one call to $\hat h$ for all $x\in[0,1]^N$;
        \item $\|\nabla h(x)\|_\infty \le \frac {e^{12}}2$, and $\|D^2 h(x)\|_{\max} \le 6e^{24}$ for all $x\in [0,1]^N$. 
    \end{enumerate}
\end{restatable}

The construction of the interpolating function $\cH(\cdot,\hat h)$ is explicit and can be found in the proof of \Cref{lem:inteprolation}. In \Cref{fig:interp}, we provide an illustrative example. 
The reason why we need a robust version of the interpolation will be discussed in the next section (and in particular in \Cref{rem:robust}).

\section{An Oracle-Efficient Reduction From \OPC to \SB} \label{sec:SB}

As already mentioned, we will also rely on the following intermediate problem, which concerns the approximation of fixed points of a continuous function. %
This problem will be used in \Cref{sec:GDA}, in combination with \OPC, to construct a query-hard instance for \GDA.

\begin{problem}[\SB]
    Given an approximation $\varepsilon>0$ and an oracle implementing a Lipschitz and $\Ccal^1$ function $F:[0,1]^d\to[0,1]^d$ and its Jacobian $J_F:[0,1]^d\to\Reals^{d\times d}$, find a $z\in[0,1]^d$ such that $\|F(z)-z\|_\infty\le\varepsilon$.
\end{problem}

The following theorem provides an oracle-efficient reduction from \OPC to \SB. %
In this construction, we rely heavily on the robustness of our interpolation as described in \Cref{lem:inteprolation}, and on the smooth step function introduced in \Cref{def:smoothstep}.

\begin{restatable}{theorem}{theoremSB} \label{thm:SB}
There is a reduction from an instance of \OPC (with oracle $\cL$) to \SB with approximation $\varepsilon=\tfrac1{12}$, such that:
\begin{itemize}[leftmargin=0.5cm]
    \item $\Ccal^\infty\ni F:[0,1]^d\rightarrow [0,1]^d$ where $d=|V|$ is the number of nodes of \OPC;
    \item any query to $F$ or $J_F$ can be simulated with at most $d=|V|$ queries to $\cL$;
    \item we can efficiently recover a \OPC solution from any \SB one;
    \item the function $F$ is $(e^{12}d)$-Lipschitz and $(12e^{24}d^2)$-smooth, and $|\partial_{z_j} F_i(z)|\le e^{12}$ for all $i,j\in[d]$.
\end{itemize}
\end{restatable}

As a corollary, since the oracles of \SB can be simulated with $\poly(d)$ queries to the oracle $\cL$ of \OPC, and we can recover a solution of \OPC from any solution of \SB, we obtain an exponential lower bound on the number of queries required by \SB, recovering a similar result to that of \citet{hirsch1989exponential}.

\begin{corollary}[Exponential query-lower bound for \SB]
For $\varepsilon=\frac{1}{12}$, \SB requires at least $2^{\Omega(d^c)}$ queries to the $F$ and $J_F$ oracles for some constant $c>0$.
\end{corollary}
It is important to make two comments regarding this result.
\begin{remark}
    We need such a stronger result on the query complexity of \SB, since the historical version of \citet*{hirsch1989exponential} only produces functions $F$ that are Lipschitz but not $\Ccal^1$. Indeed, in our final \GDA construction, the oracle implementing the gradient $\nabla f$ is simulated using $J_F$, so $J_F$ must be continuous for $f$ to be smooth.
\end{remark}

\begin{remark}\label{rem:robust}
    The robustness of the interpolation of \Cref{lem:inteprolation} is essential in the reduction from \OPC to \SB since approximate fixed points only recover inputs that are close to Boolean vertices rather than exactly Boolean.
\end{remark}

\section{Exponential Lower Bound for \GDA}\label{sec:GDA}
In this section, we combine \OPC and \SB to provide a query-hard instance of \GDA. Starting from an instance of \OPC with oracle $\cL$ on $m=|V|$ nodes, we first construct, via \Cref{thm:SB}, a \SB instance $F:[0,1]^m\to[0,1]^m$, where $m=|V|$. We then build a \GDA instance in the spirit of \citet{bernasconi2026complexity}, using many copies of \SB placed inside nodes of the \OPC instance. The major difference with respect to the proof of \citet{bernasconi2026complexity} is that we replaced the linear variational inequality used there with the \SB map $F$ as the inner problem.

Let $m$ be the size of the \SB map, i.e. $F:[0,1]^m\to [0,1]^m$, and let $\rho=\frac{1}{12}$ the approximation error. Recall that by \Cref{thm:SB} it holds $m=|V|$.

We construct the \GDA instance based on the following parameters:
\[
\delta=\tfrac{\rho^4}{400m^2e^{26}},\quad
    n=\big\lceil\tfrac{2^{13} e^{13} m^4}{\delta^3}\big\rceil,\quad\text{and}\quad
    \varepsilon=\min\big(\tfrac\delta n,\tfrac{\delta^2}{m^4 2^4 e^{14}}\big).
\]

For each $u\in V$, $i\in[n],j\in[m]$, our instance includes variables $x^{u}_{i,j}$ and $y^u_{i,j}$, so the dimension of each player is $d=|V|nm=nm^2=O(m^{12})$. Moreover we abbreviate $x^v_i=(x_{i,j}^v)_{j\in[m]}\in\Reals^m$ and $x^v:=(x_{i,j}^v)_{i\in[n],j\in[m]}\in\Reals^{mn}$ (and similarly with $y_i^v\in\Reals^m$ and $y^v\in\Reals^{mn}$).

Before defining the objective $f$, it is useful to define, using $F$, the gadget:
\begin{align*}
&H_v(x,y)=\sum_{i\in[n]} \bigg\langle F\left(\frac{x_i^v+y_i^v}{2}\right)-\frac{x_i^v+y_i^v}{2},y_i^v-x_i^v\bigg\rangle,
\end{align*}
and the ``thresholded energy function'': %
\begin{align*}%
\cE_v(x,y)=\phi_{3m,3m+1}(\|x^v-y^v\|^2),
\end{align*}
that is a smooth step function that is $0$ if $\|x^v-y^v\|^2$ is less that $3m$ and $1$ if it is bigger than $3m+1$.

Moreover, for each $w\in V$, we define a ``signal'' function $s_w(x,y)$ depending on the (unique) gate type $w$ is output of. To do so, we use specific functions depending of the gate type. 
Formally, we set
\[
s_w(x,y)=
\begin{cases}
    g(\cE_u(x,y)+\cE_v(x,y))&\text{if}\quad(u,v,w)\in\cG_{\NOR}\\
    \ell(\cE_u(x,y)-1/4)&\text{if}\quad(u,w,v)\in\cG_{\PURIFY}\\
    \ell(\cE_u(x,y)+1/4)&\text{if}\quad(u,v,w)\in\cG_{\PURIFY}\\
    \cH((\cE_{u_i}(x,y))_{i\in[N]},\cL)&\text{if}\quad(u_1,\ldots, u_N,w)\in \cG_{\ORACLE}
\end{cases}
\]
where \(
g(z)=\phi_{\nicefrac13,\nicefrac23}(1-z)\) and \(\ell(z)=\phi_{\nicefrac5{12},\nicefrac7{12}}(z)\).
Finally, we let the objective 
\begin{align}\label{eq:f}
f(x,y)=\sum_{w\in V} s_w(x,y)\cdot H_w(x,y)+\varphi(x,y),
\end{align}
where $\varphi(x,y)=\sum_{w\in V}\sum_{i\in [n]}M_i\|x_i^w-y_i^w\|^2_2$ and $M_i:=\delta(i-n/2)$.

Here, it becomes apparent the necessity of a smooth version of a query-hard Brouwer function $F$, since it is evident that the gradient of $f$ will depend on $J_F$ through the derivatives of $H_w$. 
Moreover, it is easy to see that, since all components of $f$ have bounded first and second order derivatives, thanks to \Cref{lem:basicinterp,lem:inteprolation}, then $\|f\|_\infty,\|\nabla f\|_\infty, \|D^2 f\|_\infty=\poly(d)$ and thus also $G,L,B=\poly(d)$.

Thus, to complete the proof, we just need to show that: 
\begin{enumerate}[leftmargin=0.5cm]
    \item \textbf{Correctness:} given a solution to  \GDA, we can recover a solution to either the \OPC or the \SB instance, and hence of the \OPC one by \Cref{thm:SB}.
    \item \textbf{Oracle Query Complexity:} any call to $f$ or $\nabla f$ requires at most a polynomial number of calls to the \OPC oracle $\cL$. Notice that these also include calls to the \SB oracles $F$ and $J_F$, which in turn are implemented through queries to the \OPC oracle $\cL$ by \Cref{thm:SB}.
\end{enumerate}

Indeed, if an algorithm solves \GDA in $T(d)$ oracle queries to $(f,\nabla f)$, then this would give an algorithm for \OPC with $|V|$ nodes that runs in $T(\poly(|V|))\poly(|V|)$ queries, but the \OPC lower bound of \Cref{th:OPCLB} forces this quantity to be exponential in $|V|$ and in turn implies $T(d)=2^{\Omega(d^c)}$ for some constant $c>0$.

From any solution $(x,y)$ of \GDA and define an assignment $b:V\to\{0,1,\bot\}$ as
    \begin{align}\label{eq:assignment}
    b(v)=
    \begin{cases}
        \cE_v(x,y)&\text{if}\quad \cE_v(x,y)\in \{0,1\}\\
        \bot&\text{otherwise}
    \end{cases}.
    \end{align}

Then, we can show the following dichotomy in the same spirit of \citet{bernasconi2026complexity}: either the decoded Boolean signals $b:V\to\{0,1,\bot\}$ are consistent with all gate outputs (and thus satisfy the \OPC constraints), or one of the replicated blocks $(v,i)\in V\times [n]$ already yields an approximate fixed point of the \SB instance.
Formally:

\begin{restatable}[Dichotomy]{lemma}{dichotomy} \label{lem:dicotomy}
    Let $(x,y)$ be a solution to \GDA. Then either there exists $i\in[n], v\in V$ such that $(x^v_i+y^v_i)/2$ is a solution to \SB, or $b:V\to\{0,1,\bot\}$, as defined in \Cref{eq:assignment}, is a solution to \OPC.
\end{restatable}

With this lemma, we can easily prove our main result.

\begin{theorem}\label{thm:final}
Any algorithm that has oracle access to $f: [0,1]^d \times [0,1]^d \to [-1,1]$ and to its gradient $\nabla f$ where $f$ and $\nabla f$ are $1$-Lipschitz, and outputs a $\varepsilon$-approximate stationary point with $\varepsilon\le \poly(1/d)$, requires at least $2^{d^{\Omega(1)}}$ queries to $f$ or $\nabla f$ in the worst case.
\end{theorem}
\begin{proof}
We divide the proof into two steps.
\paragraph{Correctness.} \Cref{lem:dicotomy} almost immediately proves the correctness of the reduction, as it shows that we either trivially have a solution to \SB or the \OPC constraints are satisfied by $b$, and ultimately a solution to \OPC thanks to \Cref{thm:SB}.
\paragraph{Oracle Query Complexity.}
Now we show that the oracles for $f$ and $\nabla f$ can be simulated with only polynomially many queries to the \OPC oracle $\cL$.
First, observe that evaluating $f(x,y)$ for any $x,y\in[0,1]^d$ requires only polynomially many queries to the oracle $\cL$.
Specifically, computing $s_w(x,y)$ for each $w\in V$ requires at most $|V|=\poly(d)$ queries to $\cL$ thanks to \Cref{lem:inteprolation}.
Moreover, the evaluation of the $H_w$'s terms requires at most $m^2 n= \poly (d)$ queries to $F_w(z)$. Thanks to \Cref{thm:SB}, each evaluation requires at most $m$ queries to $\cL$.
A similar analysis holds for $\nabla f$, which also requires the evaluation of $J_F$ and $\nabla s_{w}(x,y)$, for which similar bounds hold.
\end{proof}
Compared with the proof of \citet{bernasconi2026complexity}, the key new ingredient is the treatment of the large-signal regime in the dichotomy argument, \emph{i.e.}, that $s_v(x,y)=1$ implies that $b(v)=1$ or we can find a solution to \SB. The small-signal case follows the previous template, while the large-signal case now has to certify the existence of an approximate \SB fixed point, rather than a solution to the inner problem of \citet{bernasconi2026complexity}---a simpler problem related to linear variational inequalities (see \Cref{sec:largesign} for more details). We modified the inner gadget because the one based on variational inequalities was reduced from polymatrix games \citep{rubinstein2015inapproximability}. While polymatrix games are computationally intractable from a complexity perspective, they can be solved using polynomially many queries.

\begin{remark}
In contrast to \Cref{sec:SB} (see \Cref{rem:robust}), the decoded gate signals in this reduction are already thresholded to Boolean values. For correctness, we therefore only require that the interpolation matches the function on the Boolean vertices, and we do not need the robust interpolation in \Cref{lem:inteprolation}. We nevertheless use the same robust interpolation throughout for the sake of uniformity.
\end{remark}

\begin{remark}
    The construction of our function $f$ only guarantees that $G,L,B,1/\varepsilon=\poly(d)$, however, we can easily normalize $f$ by $\max(G,L,B)$, by only needing to decrease $\varepsilon$ by a polynomial factor.
\end{remark}

\subsection{Extension to High-Order Oracles}

Our main theorem is stated for first-order algorithms, which have access to oracles $f$ and $\nabla f$. However, many algorithms that use second- (or higher-)order information have been proposed for min-max optimization \citep{letcher2018stable, balduzzi2018mechanics, zhang2020newton, ha2022convergence, vyas2023beyond, chinchilla2024newton}. However, all either assume an additional special structure or offer weaker guarantees (such as local convergence).
Our argument can be extended even to rule out query-efficient algorithms of this sort. Indeed, fix any integer $p=O(1)$ and suppose that the algorithm can query $(f,\nabla f, D^2f,\ldots, D^pf)$. Then, our lower bound still applies.

The main point is that each oracle can be simulated with polynomially many queries to the underlying \OPC oracle $\cL$. The objective function $f$ is $\Ccal^\infty$ due to its construction from algebraic compositions of $\Ccal^\infty$ components. Moreover, \Cref{lem:inteprolation} can be strengthened to show that for any fixed $p$, the derivatives of the interpolation $h=\cH(\cdot,\cL)$ can be computed with just one query to $\cL$, as $h$ depends on only one vertex of the hypercube. Computing the derivative tensor $D^r f$ for any $r \le p$ also requires polynomially many queries to $\cL$ because $f$ is constructed from smooth step functions and interpolation gadgets through sums, products, and compositions, leading to polynomially many terms in the derivatives.

Therefore, if a $p$-th order algorithm solved the constructed \GDA instance using $T$ oracle queries, then we would obtain an algorithm for \OPC using $T\cdot \poly(d)$ queries to $\cL$. By the query lower bound of \Cref{th:OPCLB}, we must have that $T$ also grows exponentially.

\section{Open Problems}
Although our result resolves the query complexity of the problem in the regime where $\varepsilon = 1/\poly(d)$, the question remains open for constant values of $\varepsilon > 0$. Indeed, our lower bound does not preclude the existence of a $\poly(d)$-query algorithm for every fixed value of $\varepsilon > 0$, e.g., an algorithm making $d^{O(1/\varepsilon)}$ many queries. Proving or refuting such a PTAS-type query guarantee remains an interesting open problem. A second question concerns simpler instances. For min-max problems with a degree-$2$ objective over simplex domains, known enumeration-based techniques (e.g., the one of \citet*{lipton2003playing}) can be adapted to establish the existence of quasi-polynomial-time algorithms. %
It remains open whether this is tight: can one prove quasi-polynomial lower bounds already for degree-$2$ objectives over simplex domains, analogously to the phenomena of Nash equilibria shown by \citet*{rubinstein2017settling}?

\appendix
\newpage

\section{Proof Omitted from \Cref{sec:prelims}}\label{sec:app:prelims}

\SPlowerbound*

In order to prove this, we will use the following lower bound for finding a Brouwer fixed point.

\begin{theorem}[\citet{hirsch1989exponential}]\label{thm:HPV}
There exists a sufficiently small constant $\varepsilon > 0$ such that any algorithm that has black-box oracle access to a $2$-Lipschitz (w.r.t.\ the $\ell_\infty$ norm) function $F:[0,1]^d\to[0,1]^d$ and outputs a point $x\in[0,1]^d$ such that $\|F(x)-x\|_\infty\le\varepsilon$ requires $2^{\Omega(d)}$ queries.
\end{theorem}

\begin{proof}[Proof of \Cref{thm:SP-lower-bound}]
Let $\varepsilon$ be a sufficiently small constant for \Cref{thm:HPV}. Consider any $2$-Lipschitz function $F:[0,1]^d\to[0,1]^d$. We show how to reduce the problem of finding an $\varepsilon$-approximate fixed point of $F$ to a \SP instance.

Without loss of generality, we can assume that for all $i \in [d]$, $F_i(x) > 0$ when $x_i = 0$, and $F_i(x) < 1$ when $x_i = 1$. Indeed, if this is not the case, then we can consider the function $(1-\varepsilon/2)F + (\varepsilon/2) v$ instead of the function $F$, where $v = (1/2, \dots, 1/2) \in [0,1]^d$, and look for an $(\varepsilon/2)$-approximate fixed point.

Let $\wid := \lceil 1 + 3/\varepsilon \rceil$ and define the function $\phi: [\wid] \to [0,1], t \mapsto (t-1)/(\wid-1)$. For a vector $p \in [\wid]^\dims$, define by a slight abuse of notation $\phi(p) := (\phi(p_1), \dots, \phi(p_d))$. Note that $\phi_i(p) = \phi(p_i)$. Now define the \SP labeling $\lambda: [\wid]^\dims \to \{-1,+1\}^\dims$ as follows. For any $p \in [\wid]^\dims$
\[
\lambda(p) :=
\begin{cases}
    +1&\text{if } F_i(\phi(p)) > \phi_i(p)\\
    -1&\text{if } F_i(\phi(p)) \leq \phi_i(p)
\end{cases}
\]
Note that $\lambda$ satisfies the \SP boundary conditions. Furthermore, any query to $\lambda$ can be answered by performing at most one query to $F$.

It remains to show that any \SP solution $p^{(1)}, \dots, p^{(\dims)} \in [\wid]^\dims$ to $\lambda$ yields an $\varepsilon$-approximate fixed point of $F$. Let $q := p^{(1)}$. We will show that $\phi(q)$ is an approximate fixed point. Consider any $i \in [\dims]$. Since $p^{(1)}, \dots, p^{(\dims)}$ form a \SP solution, there exists $q'$ with $\|q-q'\|_\infty \leq 1$ such that $\lambda_i(q') = -1$. This implies that $F_i(\phi(q')) \leq \phi_i(q')$. Now we can write
$$F_i(\phi(q)) \leq F_i(\phi(q')) + |F_i(\phi(q)) - F_i(\phi(q'))| \leq \phi_i(q') + 2\varepsilon/3 \leq \phi_i(q) + \varepsilon$$
where we used the fact that $F$ is $2$-Lipschitz and $\phi$ is $(\varepsilon/3)$-Lipschitz. Using the fact that there exists $q'$ with $\|q-q'\|_\infty \leq 1$ such that $\lambda_i(q') = +1$, a similar argument also proves that $F_i(\phi(q)) \geq \phi_i(q) - \varepsilon$. Since this holds for any $i \in [\dims]$, we have shown that $\phi(q)$ is an $\varepsilon$-approximate fixed point of $F$.
\end{proof}

\section{Proof Omitted from \Cref{sec:OPC}}\label{sec:app:OPC}

\OPCLB*

We will follow the reduction from \citet[Section~3.2]{deligkas2022pure} very closely. They provide a polynomial-time reduction from the white-box version of \SP to \PC. In what follows, we describe the modifications that need to be made to their reduction in order to obtain a query-efficient reduction from the black-box version of \SP to \OPC.

Let $\lambda: [\wid]^\dims \to \{-1,+1\}^\dims$ be a \SP labeling satisfying the boundary conditions, as defined in \Cref{def:SP}. We use the same notation as in \citet[Section~3.2]{deligkas2022pure}, except that we have used $\dims$ to denote the dimension of the \SP instance, whereas they use $N$. We assume that the width $\wid$ of the instance is a sufficiently large constant such that the query lower bound of $2^{\Omega(\dims)}$ from \Cref{thm:SP-lower-bound} applies.

\paragraph{Construction.}
The construction of the \OPC instance is completely identical to the construction of the \PC instance in \citet[Section~3.2]{deligkas2022pure}, except for the ``circuit stage''. In that part of the construction, the goal is to simulate the evaluation of the \SP labeling $\lambda$ on some input. In fact, this needs to be done on $\cop$ separate occasions, where $\cop := 3 \dims \wid^2$. Fix some $k \in [\cop]$. We are given nodes $(u_{i,j}^{(k)})_{(i,j) \in [\dims] \times [\wid]}$ representing a point in $[\wid]^\dims$. Namely, for each $i \in [\dims]$, $u_{i,1}^{(k)}, \dots, u_{i,\wid}^{(k)}$ represents an element in $[\wid]$ in unary. Our task, or ``contract'', is to ensure that some designated output nodes, denoted $v_1^{(k)}, \dots, v_\dims^{(k)}$, encode the output of $\lambda$ on the input encoded by $(u_{i,j}^{(k)})_{(i,j) \in [\dims] \times [\wid]}$. By this we mean that if the $i$th output of $\lambda$ is $+1$, then $v_i^{(k)} = 1$, and if the $i$th output of $\lambda$ is $-1$, then $v_i^{(k)} = 0$. Importantly, for any given $k \in [\cop]$, the contract only requires us to ensure this correct output when \emph{all} inputs $(u_{i,j}^{(k)})_{(i,j) \in [\dims] \times [\wid]}$ are correct bits, i.e., lie in $\{0,1\}$.

In the original reduction by \citet{deligkas2022pure}, the contract is enforced by using the \PC gates to simulate the execution of the Boolean circuit computing $\lambda$. In our case, $\lambda$ is a black-box and we are not given a circuit computing it. We will thus need to use the oracle gates to enforce the contract. We let $N := \wid \dims + \dims$ and define the oracle $\cL: \{0,1\}^N \to \{0,1\}$ as follows. On input $p \in \{0,1\}^N$, decomposed as $(z,t) \in \{0,1\}^{\wid \dims} \times \{0,1\}^\dims$:
\begin{enumerate}
    \item If $t \in \{0,1\}^\dims$ contains a $1$ in exactly one position, and all other entries are $0$, then let $i$ denote the entry such that $t_i = 1$, and proceed to the next step. If this is not the case, then the oracle $\cL$ outputs something arbitrary, say $0$.
    \item Interpret $z \in \{0,1\}^{\wid \dims}$ as representing a corresponding point $\overline{z} \in [\wid]^\dims$, where each entry is given in unary. Evaluate $\lambda_i(\overline{z})$ by performing one query to $\lambda$. (Recall that $i$ is the index such that $t_i = 1$ from the first step.)
    \item If $\lambda_i(\overline{z}) = +1$, the oracle $\cL$ outputs $1$. If $\lambda_i(\overline{z}) = -1$, the oracle $\cL$ outputs $0$.
\end{enumerate}

Coming back to the contract we have to enforce, we can now use the oracle gate (with oracle $\cL$). Namely, for each $k \in [\cop]$ and $i \in [\dims]$, we ensure that $v_i^{(k)}$ has the correct value by introducing an oracle gate that has $v_i^{(k)}$ as output, and takes as input $(u_{\ell,j}^{(k)})_{(\ell,j) \in [\dims] \times [\wid]}$ for the first $\wid \dims$ bits. The remaining $\dims$ input bits are hardcoded to be $(0, \dots, 0,1,0,\dots, 0)$, where the $1$ appears in the $i$th position. This indicates to the oracle $\cL$ that we want the $i$th output of $\lambda$. In order to create these hardcoded bits, it suffices to create $d$ nodes that are guaranteed to be $0$ in any solution and $d$ nodes that are guaranteed to be $1$ in any solution. Then, each time we use an oracle gate, we can pick the corresponding hardcoded bits as input. In \Cref{fig:constantcirc} we show a constant-size gadget which ensures that a particular node always has value $0$. We include a proof of this in \Cref{lem:constantcirc} at the end of this section. From this node that always has value $0$, a node that always has value $1$ can easily be constructed using one \PURIFY and one \NOR gate.

\paragraph{Correctness.}
It is now easy to see that the contract is correctly enforced. Indeed, whenever the nodes $(u_{\ell,j}^{(k)})_{(\ell,j) \in [\dims] \times [\wid]}$ all have values in $\{0,1\}$, the fact that the oracle gate has to be satisfied will imply that the nodes $v_1^{(k)}, \dots, v_\dims^{(k)}$ encode the correct output values. When the nodes $(u_{\ell,j}^{(k)})_{(\ell,j) \in [\dims] \times [\wid]}$ do not all have values in $\{0,1\}$, then the contract does not require us to enforce anything.

As we have fulfilled the ``circuit stage'' contract, i.e., Lemma~3.3 of \citet{deligkas2022pure}, and the rest of the construction is identical, the correctness of the reduction follows by the exact same arguments as in \citet[Section~3.2]{deligkas2022pure}. Namely, it holds that any solution of the constructed \OPC instance yields a solution of the \SP instance.

\paragraph{Oracle Query Complexity.}
First, let us note that in order to answer a query to $\cL$ we only need to make at most one query to $\lambda$. Thus, by \Cref{thm:SP-lower-bound}, we know that at least $2^{\Omega(\dims)}$ queries to $\cL$ are required to solve the \OPC instance.

In order to show that this corresponds to a lower bound of $2^{\Omega(|V|^{1/3})}$, where $|V|$ is the number of nodes in the \OPC instance, it remains to argue that $|V| = O(\dims^3)$. By inspection of the construction in Section~3.2 of \citet{deligkas2022pure}, we have:
\begin{itemize}
    \item The ``purification stage'' requires $O(\wid \dims \cop) = O(\dims^2 \wid^3)$ nodes, where we used the fact that $\cop = 3 \dims \wid^2$.
    \item The ``circuit stage'', which is the only stage that we modified, requires $\dims \cop = O(\dims^2 \wid^2)$ for the oracle output nodes, as well as $O(d)$ nodes to create the hardcoded $0$ and $1$ bits.
    \item The ``sorting stage'' requires a sorting network over $\cop$ elements for each $i \in [\dims]$. By using a simple sorting network of size $O(\cop^2)$, we thus obtain an upper bound of $O(\dims \cop^2) = O(\dims^3 \wid^4)$ nodes.
    \item The ``selection stage'' does not introduce any additional nodes.
\end{itemize}
Putting everything together, we have $|V| = O(\dims^3 \wid^4) = O(\dims^3)$, where we used the fact that $M$ is a constant. A subtle point is that the reduction in Section~3.2 of \citet{deligkas2022pure} also uses gates \textsc{And}, \textsc{Or}, \textsc{Not}, and \textsc{Copy}. However, it is easy to see that each of those gates can be simulated by a constant-size gadget that uses only \NOR and \PURIFY gates. Thus, we still have $|V| = O(\dims^3)$. This completes the proof of \Cref{th:OPCLB}.

\begin{remark}
By using a more involved sorting network of size $O(\cop \log \cop)$, instead of $O(\cop^2)$, one can obtain an instance of size $|V| = O(d^2 \log d)$ \citep{ajtai19830}. Thus, the lower bound in \Cref{th:OPCLB} can be improved to $2^{\widetilde{\Omega}(\sqrt{|V|})}$, where $\widetilde{\Omega}$ allows for division by polylogarithmic factors. We leave open whether this lower bound can be further improved.
\end{remark}

Below, we include a lemma which proves the correctness of the gadget that we used to create the hardcoded bits.

\begin{figure}[!tp]
    \centering
    \scalebox{0.75}{\def\gateW{14mm}
\def\gateH{9.5mm}
\def\portSep{\gateH/4}
\def\sepone{\gateW/2}
\def\septwo{\sepone+\gateW/2}
\def\lowerDrop{1.35cm}

\begin{tikzpicture}[
  font=\sffamily,
  >=Latex,
  wire/.style={
    -{Latex[length=2.2mm]},
    line width=0.75pt,
    rounded corners=5pt,
    line cap=round,
    line join=round
  },
  plain wire/.style={
    line width=0.75pt,
    rounded corners=5pt,
    line cap=round,
    line join=round
  },
  gate/.style={
    draw,
    thick,
    rounded corners=4pt,
    minimum width=\gateW,
    minimum height=\gateH,
    font=\small\sffamily
  },
  pur/.style={
    gate,
    draw=blue!55!black,
    fill=blue!8
  },
  nor/.style={
    gate,
    draw=green!45!black,
    fill=green!10
  },
  nodevar/.style={
    circle,
    draw=black!65,
    fill=white,
    thick,
    minimum size=3.8mm,
    inner sep=0pt,
    font=\scriptsize
  },
  copynode/.style={
    circle,
    draw=black!55,
    fill=white,
    thick,
    minimum size=6mm,
    inner sep=0pt,
    font=\scriptsize
  }
]

\node[pur]     (Pa) at (0,0) {\PURIFY};
\coordinate (PaIn)  at (Pa.west);
\coordinate (PaO1)  at ($(Pa.east)+(0,\portSep)$);
\coordinate (PaO2)  at ($(Pa.east)+(0,-\portSep)$);

\node[nodevar] (aone)  at ($(Pa.east)+(\sepone,\portSep)$) {$v_2$};
\node[nodevar] (atwo)  at ($(Pa.east)+(\sepone,-\portSep)$) {$v_3$};
\coordinate (a12_)  at ($(aone)+(0,-\portSep)$);

\node[nor]     (Nb) at ($(a12_)+(\septwo,0)$) {\NOR};
\node[nodevar] (b)  at ($(Nb.east)+(\sepone,0)$) {$v_4$};
\coordinate (NbI1)  at ($(Nb.west)+(0,\portSep)$);
\coordinate (NbI2)  at ($(Nb.west)+(0,-\portSep)$);
\coordinate (NbOut) at (Nb.east);
\coordinate (b_)  at (b);

\node[pur]     (Pb) at ($(b_)+(\septwo,0)$) {\PURIFY};
\node[nodevar] (c)  at ($(Pb.east)+(\sepone,\portSep)$) {$v_5$};
\coordinate (c_)  at (c);

\node[pur]     (Pc) at ($(c)+(2*\septwo,-\lowerDrop)$) {\PURIFY};

\node[nodevar] (cone)  at ($(Pc.east)+(\sepone,\portSep)$) {$v_6$};
\node[nodevar] (ctwo)  at ($(Pc.east)+(\sepone,-\portSep)$) {$v_7$};%
\coordinate (c_)  at ($(cone)+(0,-\portSep)$);

\node[nor]     (Nd) at ($(c_)+(\septwo,0)$) {\NOR};
\node[nodevar] (d)  at ($(Nd.east)+(\sepone,0)$) {$v_8$};

\node[nor]     (Ne) at ($(d)+(2*\septwo,\lowerDrop)$) {\NOR};
\node[nodevar] (e)  at ($(Ne.east)+(\sepone,0)$) {$v_9$};

\node[nodevar] (a)  at ($(Nb)+(0,-\lowerDrop+\portSep)$) {$v_1$};

\coordinate (PbIn)  at (Pb.west);
\coordinate (PbO1)  at ($(Pb.east)+(0,\portSep)$);
\coordinate (PbO2)  at ($(Pb.east)+(0,-\portSep)$);

\coordinate (PcIn)  at (Pc.west);
\coordinate (PcO1)  at ($(Pc.east)+(0,\portSep)$);
\coordinate (PcO2)  at ($(Pc.east)+(0,-\portSep)$);

\coordinate (NdI1)  at ($(Nd.west)+(0,\portSep)$);
\coordinate (NdI2)  at ($(Nd.west)+(0,-\portSep)$);
\coordinate (NdOut) at (Nd.east);

\coordinate (NeI1)  at ($(Ne.west)+(0,\portSep)$);
\coordinate (NeI2)  at ($(Ne.west)+(0,-\portSep)$);
\coordinate (NeOut) at (Ne.east);

\coordinate (Paleft) at ($(Pa.west)+(-0.65cm,-\lowerDrop/2)$);
\draw[wire] (a.west) -| (Paleft) |- (Pa.west);

\draw[wire] (PaO1) -- (aone.west);
\draw[wire] (aone.east) -- (NbI1);

\draw[wire] (PaO2) -- (atwo.west);
\draw[wire] (atwo.east) -- (NbI2);

\draw[wire] (NbOut) -- (b.west);

\draw[wire] (b.east) -- (PbIn);

\draw[wire] (PbO1) -- (c.west);
\draw[wire] (PbO2) -- ++(0.35cm,0) |- (a.east);

\coordinate (cBus) at ($(c.east)+(0.55cm,0)$);

\draw[plain wire] (c.east) -- (cBus);
\draw[wire] (cBus) |- (PcIn);
\draw[wire] (cBus) |- (NeI1);

\draw[wire] (PcO1) -- (cone.west);
\draw[wire] (cone.east) -- (NdI1);

\draw[wire] (PcO2) -- (ctwo.west);
\draw[wire] (ctwo.east) -- (NdI2);

\draw[wire] (NdOut) -- (d.west);

\draw[wire] (d.east) -- ++(0.35cm,0) |- (NeI2);
\draw[wire] (NeOut) -- (e.west);

\end{tikzpicture}}
    \caption{The gadget implementing the constant node $b(v_9)=0$. Assuming that this gadget actually forces $b(v_9)$ (proven in \Cref{lem:constantcirc}), we can easily implement a similar one that forces an assignment of $1$ to a specific node. Indeed, we could further apply a \PURIFY gate to $v_9$, with outputs $v_{10}$ and $v_{11}$, and then a further \NOR gate with inputs $v_{10}$ and $v_{11}$ and output $v_{12}$. If $b(v_9)=0$ it is clear that only $b(v_{12})=1$ satisfies the additional gates.}
    \label{fig:constantcirc}
\end{figure}
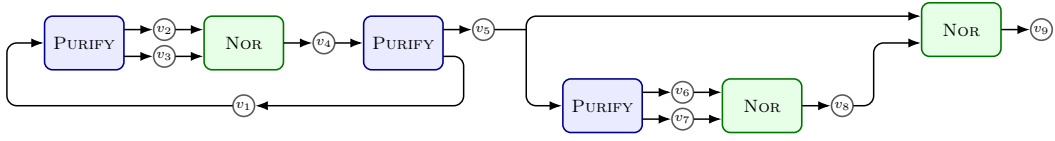

\begin{lemma}\label{lem:constantcirc}
    In the gadget described in \Cref{fig:constantcirc}, for every assignment $b:V\to\{0,1,\bot\}$ that satisfies all gates, it must hold that $b(v_9)=0$.
\end{lemma}
\begin{proof}
    Let $b:V\to\{0,1,\bot\}$ be any assignment that satisfies all gates. First, we will argue that $b(v_1)=\bot$. Assume by contradiction that $b(v_1)=1$, then we would also have $b(v_2)=b(v_3)=1$ by definition of the \PURIFY gate. Then $b(v_4)=0$ by definition of the \NOR constraints. Now the \PURIFY gate, with $v_4$ as input, has a node with a pure bit input, and thus its outputs are both pure and assigned with the same bit, which is $0=b(v_4)=b(v_5)=b(v_1)$. Contradicting the fact that $b(v_1)$ was assumed to be $1$. A similar, symmetric situation will occur if we assume $b(v_1)=0$, thus leaving only $b(v_1)=\bot$ as an option. 

    Now we claim that $b(v_5)\in\{0,1\}$. Indeed, the \PURIFY gate with input $v_4$ has as outputs $v_1$ and $v_5$, but the \PURIFY gate constraints force at least one of the two output bits to be pure, and, since $v_1$ cannot be pure, then $v_5$ must be.

    Now observe that if $b(v_5)\in\{0,1\}$, then the statement follows. Indeed, if $b(v_5)=0$, then $b(v_6)=b(v_7)=b(v_5)=0$ by the constraints of the \PURIFY gate, and $b(v_8)=1$ by the ones on the \NOR gate. Then, if $b(v_5)=0$ and $b(v_8)=1$, then the last \NOR gate forces $b(v_9)=0$.

    Symmetrically, if $b(v_5)=1$ then it leads to $b(v_8)=0$ and similarly to before, to $b(v_9)=0$, by the last \NOR gate constraints.
\end{proof}

\section{Proofs Omitted from \Cref{sec:smooth}}\label{sec:app1}
\basicinterp*
\begin{proof}
We only show the bounds on the derivatives. Define $a:=x-c_1$ and $b:=c_2-x$, and also $A=e^{-1/a}$, $B=e^{-1/b}$, $S=A+B$ and $\Delta=c_2-c_1$. Then we for all $x\in(c_1,c_2)$ we can write $\phi_{c_1,c_2}(x)=A/S$. We are only interested in the behavior of $\phi_{c_1,c_2}$ in $x\in(c_1,c_2)$ since outside this interval $\phi_{c_1,c_2}'=\phi_{c_1,c_2}''=0$. We can observe that $A'=\frac A{a^2}$, $B'=\frac B{b^2}$, and thus $|A'|,|B'|\le1$ since $\frac{e^{-1/t}}{t^2}\le1$ for $t>0$.
Moreover, (since at least one between $a$ and $b$ will be at least $\Delta/2$) at least one between $A$ and $B$ is at least $e^{-2/\Delta}$ so $S\ge e^{-2/\Delta}$.

Thus, we can bound the first derivative by observing that $\phi'_{c_1,c_2}(x)=\frac{A'(A+B)-A(A'+B')}{S^2}=\frac{A'B-AB'}{S^2}$, which implies that
\begin{align*}
|\phi'_{c_1,c_2}(x)|&\le \frac{|A'|B+A|B'|}{S^2}\\
&\le \frac{A+B}{S^2}\\
&\le \frac1S\le  e^{2/\Delta}.
\end{align*}

For the second derivative, we can observe that 
\begin{align*}
    \phi''_{c_1,c_2}(x)&=\frac{(A'B-AB')'S^2-(S^2)'(A'B-AB')}{S^4}\\
    &=\frac{(A''B-AB'')S^2-2S(A'+B')(A'B-AB')}{S^4}\\
    &=\frac{A''B-AB''}{S^2}-2\frac{(A'+B')(A'B-AB')}{S^3},
\end{align*}
Moreover, we can observe that, for $t>0$ we have $\frac{e^{-1/t}}{t^4}\le5$ and $\frac{e^{-1/t}}{t^3}\le\frac32$, so that $|A''|\le e^{-1/a}\left(\frac{1}{a^4}+\frac2{a^3}\right)\le 8$ and the same holds for $|B''|$. Plugging these bound in $|\phi''_{c_1,c_2}(x)|$ we obtain
\begin{align*}
|\phi''_{c_1,c_2}(x)|&\le \frac{|A''|B+A|B''|}{S^2}+2\frac{(|A'|+|B'|)(|A'|B+A|B'|)}{S^3}\\
&\le \frac{8}{S}+\frac{4}{S^2}\\
&\le 8e^{-2/\Delta}+4e^{-4/\Delta}\\
&\le 12 e^{-4/\Delta},
\end{align*}
concluding the proof.
\end{proof}

\interpolation*
\begin{proof}
Let 
\(
    \alpha(t)=1-\phi_{\frac16,\frac13}(t).%
\)
and for any vertex $y\in\{0,1\}^N$ define 
\[
    \Phi_{y}(x)=\prod_{i\in[N]}\alpha(y_i+(1-2y_i)x_i).
\]
Crucially, by \Cref{lem:basicinterp}, for a vertex $y\in\{0,1\}^N$, we have that $\Phi_{y}(x)=1$ for all $x \in B_{\nicefrac16}^\infty(y)$; that $\Phi_{y}(x)=0$ for all $x \notin B_{\nicefrac13}^\infty(y)$; and that $\Phi_{y}(x)\in(0,1)$ for all $x \in B_{\nicefrac13}^\infty(y)\setminus B_{\nicefrac16}^\infty(y)$.

For any $x\in[0,1]^N$ we call any vertex $y\in\{0,1\}^N$ \emph{active} if $x\in B_{1/3}^\infty(y)$. Note that, since $\ell_\infty$-balls of radius $1/3$ centered at the vertices of the hypercube are disjoint, at most $1$ vertex is active for all $x\in[0,1]^N$. So if there is any coordinate $i$ such that $x_i\in(\tfrac13,\tfrac23)$, then there is no active vertex. Otherwise, we can identify the active vertex $y^*$ by rounding every $x_i\le1/3$ to $y^*_i=0$ and every $x_i\ge2/3$ to $y^*_i=1$.

Then we define 
\[
h(x)=\frac{1}{2}+\sum_{y\in\{0,1\}^N}\Phi_{y}(x)(\hat h(y)-\tfrac12)
\]
for all $x\in[0,1]^N$.
Since there is at most one active node for all $x\in[0,1]^N$, we have that either $\Phi_{y}(x)=0$ for all $y\in\{0,1\}^N$ or that there is a single active one $y^*$. In the first case, we have $h(x)=\frac12$, in the second one we have $h(x)=\frac{1}{2}+\Phi_{y^*}(x)(\hat h(y^*)-1/2)$. Moreover if $x\in B_{\nicefrac16}^\infty(y^*)$ we have that $\Phi_{y^*}(x)=1$ and thus $h(x)=\hat h(y^*)$. This proves the first and second items of the statement. Note that since we can know what the unique active node (if any) is, we need only one evaluation of $\hat h$.

By direct calculations we have that for all $y\in\{0,1\}^N$ we have:
\begin{itemize}
    \item $\partial_{x_j}\Phi_y(x)=(1-2y_j)\alpha'(y_j+(1-2y_j)x_j)\prod_{i\neq j}\alpha(y_i+(1-2y_i)x_i)$;
    \item $\partial_{x_j,x_k}\Phi_y(x)=(1-2y_j)\alpha'(y_j+(1-2y_j)x_j)(1-2y_k)\alpha'(y_k+(1-2y_k)x_k)\prod_{i\neq j,k}\alpha(y_i+(1-2y_i)x_i)$ if $k\neq j$;
    \item $\partial_{x_j,x_k}\Phi_y(x)=\alpha''(y_j+(1-2y_j)x_j)\prod_{i\neq j}\alpha(y_i+(1-2y_i)x_i)$;
\end{itemize}
which in turn, thanks to \Cref{lem:basicinterp} and $\alpha(x)=\phi_{\tfrac16,\tfrac13}(\tfrac12-x)$, implies that for all $x\in[0,1]^N$:
\begin{itemize}
    \item $|\partial_{x_j}\Phi_y(x)|\le\|\alpha'\|_\infty\le e^{12}$;
    \item $|\partial_{x_j,x_k}\Phi_y(x)|\le\|\alpha'\|_\infty^2\le e^{24}$ if $k\neq j$;
    \item $|\partial_{x_j,x_k}\Phi_y(x)|\le\|\alpha''\|_\infty\le12 e^{24}$.
\end{itemize}

Thus, thanks to the observations above, we can also bound the derivative. A direct calculation shows that
\[\partial_{x_j} h(x)=
\begin{cases}
    0&\text{if no vertex is active,}\\
    (\hat h(y)-\tfrac12)\partial_{x_j} \Phi_y(x)&\text{if $y$ is active.}
\end{cases}
\]
and thus,
\[
|\partial_{x_j}h(x)|\le \frac{1}{2}|\partial_{x_j} \Phi_y(x)| \le\frac12 \|\alpha'\|_\infty\le e^{12}/2.
\]
Similarly, we can note that 
\[\partial_{x_j,x_k} h(x)=
\begin{cases}
    0&\text{if no vertex is active,}\\
    (\hat h(y)-\tfrac12)\partial_{x_j,x_k} \Phi_y(x)&\text{if $y$ is active.}
\end{cases}
\]
and thus, thanks to the previous observations, $|\partial_{x_j,x_k} h(x)|\le 6e^{24}$.
\end{proof}

\section{Proof Omitted from \Cref{sec:SB}}

\theoremSB*
\begin{proof}
Given an instance of \OPC with vertex set $V$ and oracle $\cL$, we build a smooth function $F:[0,1]^d\to[0,1]^d$, where as promised $d=|V|$. 
Let $\varepsilon=1/12$. For every $w\in V$, the $w$-th component $F_w:[0,1]^d\to[0,1]$ is defined based on the gate type for which $w$ is an output:

\begin{itemize}[leftmargin=0.5cm]
    \item If $(u,v,w)\in \cG_{\NOR}$,  we define $F_w(z)=g(z_u+z_v)$, where
    \begin{align*} %
    g(z)=1-\phi_{\tfrac13,\tfrac23}(z).%
    \end{align*}
    \item If $(u,v,w)\in \cG_{\PURIFY}$, we define $F_v(z)=\ell(z_u+\tfrac14)$ and $F_w(z)=\ell(z_u-\tfrac14)$, where
    \begin{align*} %
    \ell(z)=\phi_{\tfrac5{12},\tfrac7{12}}(z).%
\end{align*}
\item If $(u_1,\ldots,u_N,v)\in \cG_{\ORACLE}$, we define $F_v(z)=\cH((z_{u_1},\ldots,z_{u_N}), \cL)$ according to \Cref{lem:inteprolation}.
\end{itemize}

\paragraph{Correctness.}
Let $z\in[0,1]^d$ be any point satisfying $\|F(z)-z\|_\infty\le\varepsilon$. We define the assignment $b:V\rightarrow \{0,1,\bot\}$ as follows:
\[b(v)=\begin{cases}
    0&\text{if}\quad z_v\le \tfrac16\\
    \bot&\text{if}\quad z_v\in(\tfrac16,\tfrac56)\\
    1&\text{if}\quad z_v\ge \tfrac56
\end{cases}.\]

Recall that by \Cref{lem:basicinterp}, it holds $g(z)=0$ if $z\ge \frac23$, $g(z)\in (0,1)$ if $z\in (\frac13,\frac23)$, and $g(z)=1$ if $z\le \frac{1}{3}$.
Moreover, $\ell(z)=0$ if $z\le \frac{5}{7}$, $\ell(z)\in (0,1)$ if $z\in (\frac{5}{12},\frac{7}{12})$, and $\ell(z)=1$ if $z\ge \frac{7}{12}$.

We prove that $b$ is a valid assignment to $\OPC$ since it satisfies all the gate constraints posed by $\OPC$:

\begin{itemize}[leftmargin=0.5cm]
\item  Gate $(u,v,w)\in\cG_{\NOR}$.
If $b(v)=b(u)=0$, then by definition $z_u,z_v\le 1/6$ and $z_u+z_v\le1/3$ and thus $F_w(z)=g(z_u+z_v)=1$. Since $z$ is a solution to \SB it holds that $z_w\ge 1-\varepsilon> 5/6$ and thus $b(w)=1$ as required. 
Similarly, if $b(u)=1$ or $b(v)=1$, then we have that $z_u+z_v\ge 5/6$ and thus $F_w(z)=g(z_u+z_v)=0$, which implies that $z_w\le\varepsilon<1/6$ and thus $b(w)=0$.
In all other cases \OPC does not impose any constraints on the \NOR gate.

\item Gate $(u,v,w)\in\cG_{\PURIFY}$.
If $b(u)=0$ then $z_u\le1/6$. Thus, $F_v(z)=\ell(z_u+1/4)=0$. Moreover, $z_u-1/4\le 5/12$ and thus also $F_w(z)=0$. Thus, since $z$ is a solution to \SB, $z_v,z_w\le \varepsilon < 1/6$ and $b(v)=b(w)=0$.
Similarly, if $b(u)=1$ then $z_u\ge 5/6$, $z_u+1/4,z_u-1/4\ge7/12$ and $F_v(z)=F_w(z)=1$. It follows that $z_v,z_w\ge 1-\varepsilon>5/6$ and $b(v)=b(w)=1$.
Finally, if $b(u)=\bot$ then $z_u\in(1/6,5/6)$ and it is easy to see that either $z_u+1/4\ge 7/12$ or $z_u-1/4\le 5/12$. In the first case, we have $F_v(z)=1$ (and thus $z_v\ge 1-\varepsilon>5/6$ and $b(v)=1$), while in the second case we have $F_w(z)=0$ (and thus $z_w\le \varepsilon<1/6$ and $b(w)=0$).
This shows that the \PURIFY constraints are also satisfied.

\item Gate $(u_1,\ldots, u_N,v)\in\cG_{ORACLE}$.
Assume $(b(u_1),\ldots, b(u_N))\in\{0,1\}^N$, otherwise there are no constraints. %
Define $z'=(b(u_1),\ldots, b(u_N))$.
Note that if $z'_{i}=b(u_i)=1$ then $z_{u_i}\ge 5/6$ and $|z'_{i}-z_{u_i}|\le1/6$.
Similarly, if $z'_{i}=b(u_i)=0$, then  $z_{u_i}\le 1/6$ and $|z'_{i}-z_{u_i}|\le 1/6$.
Therefore, $(z_{u_1},\ldots, z_{u_N})\in B^\infty_{1/6}(z')$ and from \Cref{lem:inteprolation} this implies that $F_v(z)=\cL(z')$ and thus $|z_v-F_v(z)|\le\varepsilon$.
If $\cL(z')=1$ then $z_v\ge 1-\varepsilon>5/6$ and $b(v)=1$. 
If $\cL(z')=0$ then $z_v\le\varepsilon<1/6$ and $b(v)=0$. In either case $\cL(z')=b(v)$. This shows that the constraints on the \ORACLE gates are also satisfied.
\end{itemize}

\paragraph{Oracle Query Complexity.}
We show that any evaluation of $F$ or its Jacobian $J_F$ requires at most $|V|$ queries to the oracle $\cL$. Only coordinates corresponding to outputs of \ORACLE gates require access to $\cL$. Since there are at most $|V|$ such components, we can compute the value and the Jacobian in at most $|V|$ oracle calls to $\cL$, since \Cref{lem:inteprolation} shows that each component value and derivative can be obtained via at most one query to $\cL$.

\paragraph{Bounding the Lipschitzness and smoothness.}

Each coordinate $w\in V$ of $F_w$, is either the output of $g(z_u+z_v),\ell(z_u\pm1/4)$ or $\cH((z_{u_1},\ldots,z_{u_N}),\cL)$. Furthermore, by and so by \Cref{lem:basicinterp} we can bound $\|\partial_{z_u} g\|_\infty\le e^6, \|\partial_{z_u}\ell\|_\infty\le e^{12}$ for all $u\in V$, while \Cref{lem:inteprolation} ensures that $\|\partial_{z_u}\cH((z_{u_1},\ldots,z_{u_N}),\cL)\|_\infty\le e^{12}/2$ for all $u\in V$. In all cases, we have
\(
|\partial_{z_u}F_w(z)|\le e^{12},
\)
and thus every component is $e^{12}\sqrt{d}$-Lipschitz making the function $(e^{12}d)$-Lipschitz.
Similarly, we can show that $F$ is polynomially smooth. Indeed, by the same reasoning above, we have that $\|\partial_{z_u,z_{u'}} g\|_\infty\le 12e^{12}, \|\partial_{z_u,z_{u'}}\ell\|_\infty\le 12e^{24}$ for all $u,u'\in V$ by \Cref{lem:basicinterp}, while $\|\partial_{z_u,z_{u'}}\cH((z_{u_1},\ldots,z_{u_N}),\cL)\|_\infty\le6e^{24}$ for all $u,u'\in V$ by \Cref{lem:inteprolation}. This lets us conclude that $|\partial_{z_u,z_{u'}}F_w(z)|\le 12e^{24}$, and thus, that $F$ is $(12e^{24}d^2)$-smooth.
\end{proof}

\section{Proof Omitted from \Cref{sec:GDA}}

We start by computing the partial derivatives of $f$.
To do so, we first need some additional notation. We denote with $\In(w)\subseteq V$ the set of nodes that are inputs to the gate to which $w$ is an output, and $\Out(q)=\{w:q\in \In(w)\}$ is the set of nodes that are outputs of a gate to which $q$ is an input. Moreover, we define the displacement function $G(z)=F(z)-z$, where $F$ is the function defining the \SB instance.

\begin{lemma}\label{lem:decomp}
The partial derivatives of f (defined as per \Cref{eq:f}) are
\begin{subequations}
\begin{align}
    &\partial_{x_{i,j}^q} f(x,y)=s_q(x,y)\partial_{x_{i,j}^q} H_q(x,y)+2(M_i+\Delta_q(x,y))(x_{i,j}^q-y_{i,j}^q)\label{eq:derivx}\\   
    &\partial_{y_{i,j}^q} f(x,y)=s_q(x,y)\partial_{y_{i,j}^q} H_q(x,y)-2(M_i+\Delta_q(x,y))(x_{i,j}^q-y_{i,j}^q)\label{eq:derivy},
\end{align}
\end{subequations}
where 
\begin{align}\label{eq:defDelta}
    & \Delta_q(x,y)=\sum_{w\in\Out(q)} H_w(x,y)\phi'_{3m,3m+1}(\|x^q-y^q\|^2)\partial_{\cE_q(x,y)} s_w(x,y).
\end{align}
\end{lemma}

\begin{proof}
    We recall the definition of $f(x,y)$ which is:
    \[
    f(x,y)=\sum_{w\in V} s_w(x,y)\cdot H_w(x,y)+\varphi(x,y).
    \]
    Since $H_w$ is a local term (i.e., it depends only on the variables $x^w,y^w$), and $s_w$ depends only on the nodes that are inputs of $w$ (and in particular $q\notin \In(q)$), we have that:
    \begin{align}\label{eq:tmpsum}
    \partial_{x_{i,j}^q}\left(\sum_{w\in V} s_w(x,y)\cdot H_w(x,y)\right)=s_q(x,y)\partial_{x_{i,j}^q}H_q(x,y)+\hspace{-0.2cm}\sum_{w\in\Out(q)} \hspace{-0.2cm}H_w(x,y)\cdot \partial_{x_{i,j}^q} s_w(x,y).
    \end{align}
    Then, by chain rule, for every $w\in \Out(q)$, we have $\partial_{x_{i,j}^q} s_w(x,y)=\partial_{\cE_q(x,y)}s_w(x,y)\partial_{x_{i,j}^q}\cE_q(x,y)$,while for every $w\not\in\Out(q)$ we have $\partial_{x_{i,j}^q} s_w(x,y)=0$.

    Note that by definition, $\cE_q(x,y)$ is a function of only $\|x^q-y^q\|^2$ and in particular $\cE_q(x,y)=\phi_{3m,3m+1}(\|x^q-y^q\|^2)$. Thus, for all $w\in\Out(q)$:
    \begin{align}\label{eq:tmpsw}
    \partial_{x_{i,j}^q}s_w(x,y)=2\partial_{\cE_q(x,y)} s_w(x,y)\phi'_{3m,3m+1}(\|x^q-y^q\|^2)(x_{i,j}^q-y_{i,j}^q).
    \end{align}
    Moreover, it is easy to check that:
    \begin{align}\label{eq:tmpM}
    \partial_{x_{i,j}^q}\varphi(x,y) = 2M_i(x_{i,j}^q-y_{i,j}^q).
    \end{align}

    Combining \Cref{eq:tmpsum,eq:tmpsw,eq:tmpM} with the definition of $\Delta_q(x,y)$ of \Cref{eq:defDelta}, we can deduce that:
    \[
    \partial_{x_{i,j}^q} f(x,y)=s_q(x,y)\partial_{x_{i,j}^q}H_q(x,y)+2\left(\Delta_q(x,y)+M_i\right)(x_{i,j}^q-y_{i,j}^q).
    \]

    Similarly, we get that
    \[
    \partial_{y_{i,j}^q} f(x,y)=s_q(x,y)\partial_{y_{i,j}^q}H_q(x,y)-2\left(\Delta_q(x,y)+M_i\right)(x_{i,j}^q-y_{i,j}^q),
    \]
    where the only difference comes from the derivative of $\|x^q-y^q\|^2$, which gives $\partial_{x_{i,j}^q}\|x^q-y^q\|^2=-\partial_{y_{i,j}^q}\|x^q-y^q\|^2=2(x_{i,j}^q-y_{i,j}^q)$.
\end{proof}

Then, we prove the following technical lemma that will help us compute an upper bound on the partial derivatives of $H_q$.

\begin{lemma}\label{lem:derivH}
    For every $(i,j,q)\in[n]\times[m]\times V$ we have
    \begin{subequations}
        \begin{align}
         & \partial_{x_{i,j}^q} H_q(x,y)= -G_j\left(\frac{x_i^q+y_i^q}{2}\right)+ R_{i,j,q}(x,y),\\
        &\partial_{y_{i,j}^q} H_q(x,y)= G_j\left(\frac{x_i^q+y_i^q}{2}\right)+ R_{i,j,q}(x,y),
        \end{align}
    \end{subequations}
    where $R_{i,j,q}(x,y):=\frac{1}{2}\sum_{k\in[m]}  (y_{i,k}^q-x_{i,k}^q)\cdot \partial_{z_j}G_k\left(\frac{x_i^q+y_i^q}{2}\right)$ and 
    Moreover, it holds that $\max_{i,j,q}|R_{i,j,q}(x,y)|\le e^{12}\|x_i^q-y_i^q\|_1$ and $\max_{i,j,q}\{|\partial_{x_{i,j}^q} H_q(x,y)|,|\partial_{y_{i,j}^q} H_q(x,y)|\}\le m e^{13}$.
\end{lemma}
\begin{proof}
     Recall that we have
    \[
    H_q(x,y)= \sum_{\ell \in [n]} \Bigg\langle G\left(\frac{x_\ell^q+y_\ell^q}{2}\right), y_\ell^q-x_\ell^q\Bigg\rangle. 
    \]
    Since $x_{i,j}^q$ appears only in the $\ell=i$ summand we get
    \begin{align*}
    \partial_{x_{i,j}^q} H_q(x,y)&= -G_j\left(\frac{x_i^q+y_i^q}{2}\right) +  \sum_{k \in [m]}(y_{i,k}^q-x_{i,k}^q) \partial_{x_{i,j}^q} G_k\left(\frac{x_i^q+y_i^q}{2}\right)\\
    &= -G_j\left(\frac{x_i^q+y_i^q}{2}\right) + \frac12 \sum_{k \in [m]}(y_{i,k}^q-x_{i,k}^q)  \partial_{z_j}G_k\left(\frac{x_i^q+y_i^q}{2}\right)\\
    &=-G_j\left(\frac{x_i^q+y_i^q}{2}\right)+R_{i,j,q}(x,y). 
    \end{align*}

    Similarly, we can conclude that%
    \[
    \partial_{y_{i,j}^q} H_q(x,y)=G_j\left(\frac{x_i^q+y_i^q}{2}\right)+R_{i,j,q}(x,y).
    \]
    
    From the bounds of \Cref{thm:SB} it is clear that
    \begin{align*}
    \max_{i,j,q}|R_{i,j,q}(x,y)|&\le\frac{1}{2}\max_{k\in[m]}\left|\partial_{z_j} G_k\left(\frac{x_i^q+y_i^q}{2}\right)\right|\cdot\|x_i^q-y_i^q\|_1\\
    &=\frac{1}{2}\max_{k\in[m]}\left|\partial_{z_j} F_k\left(\frac{x_i^q+y_i^q}{2}\right)-\mathbb{I}(j=k)\right|\cdot\|x_i^q-y_i^q\|_1\\
    &\le e^{12}\|x_i^q-y_i^q\|_1.
    \end{align*}
    This easily implies also that
    \begin{align*}
    \max_{i,j,q}\{|\partial_{x_{i,j}^q} H_q(x,y)|,|\partial_{y_{i,j}^q} H_q(x,y)|\}
    &\le \left|G_j\left(\frac{x_i^q+y_i^q}{2}\right)\right|+| R_{i,j,q}(x,y)|\\
    &\le 1+e^{12}\|x_i^q-y_i^q\|_1\\
    &\le m e^{13},
    \end{align*}
    concluding the proof.
\end{proof}

Now we show that when the regularizer $M_i$ does not guess correctly the noise term $\Delta_q(x,y)$, the players are forced to play close to each other, namely:

\begin{lemma}\label{lem:xminusy}
    Consider any solution $(x,y)$ to \GDA and any $q\in V$, $i \in [n]$, $j \in [m]$. If $M_i+\Delta_q(x,y)\neq 0$, it holds that:
    \[
    |x_{i,j}^q-y_{i,j}^q|\le  me^{13}\frac{s_q(x,y) }{|M_i+\Delta_q(x,y)|}+\sqrt{\frac{\varepsilon}{|M_i+\Delta_q(x,y)|}}.
    \]
\end{lemma}
\begin{proof}
    Consider the optimality conditions of the $x$-player w.r.t.~$y_{i,j}^q$ and viceversa, i.e.
    \begin{subequations}
    \begin{align}
    -&\partial_{x_{i,j}^q}f(x,y)(y_{i,j}^q-x_{i,j}^q)\le\varepsilon \label{eq:opt1}\\
    &\partial_{y_{i,j}^q}f(x,y)(x_{i,j}^q-y_{i,j}^q)\le\varepsilon.\label{eq:opt2}
    \end{align}
    \end{subequations}

    Define $A:=M_i+\Delta_q(x,y)$.
    If $A> 0$, substituting the expression of
 $\partial_{x_{i,j}^q}f(x,y)$ of \Cref{eq:derivx} in \Cref{eq:opt1}, 
    we obtain
    \[
    -\left(s_q(x,y)\partial_{x_{i,j}^q} H_q(x,y)+2(M_i+\Delta_q(x,y))(x_{i,j}^q-y_{i,j}^q)\right)(y_{i,j}^q-x_{i,j}^q)\le\varepsilon.
    \]

    On the other hand, if $A<0$, substituting the expression of
 $\partial_{y_{i,j}^q}f(x,y)$ of \Cref{eq:derivy} in \Cref{eq:opt2} we obtain
    \[
    \left(s_q(x,y)\partial_{y_{i,j}^q} H_q(x,y)-2(M_i+\Delta_q(x,y))(x_{i,j}^q-y_{i,j}^q)\right)(x_{i,j}^q-y_{i,j}^q)\le\varepsilon.
    \]

    In both cases, rearranging, implies the following quadratic expression in $|\kappa|:=|x_{i,j}^q-y_{i,j}^q|$:
    \begin{align}\label{eq:quadratic}
    2\kappa^2 |A|\le\varepsilon+|\kappa| C s_q(x,y),
    \end{align}
    where by \Cref{lem:derivH} we have $C:=\max_{i,j,q}\{|\partial_{x_{i,j}^q} H_q(x,y)|,|\partial_{y_{i,j}^q} H_q(x,y)|\}\le me^{13}$.
    Solving \Cref{eq:quadratic} implies that
    \begin{align*}
    |\kappa|&\le\frac{Cs_q(x,y)+\sqrt{C^2s_q^2(x,y)+8|A|\varepsilon}}{4|A|}\\
    &\le me^{13}\frac{s_q(x,y)}{|M_i+\Delta_q(x,y)|}+\sqrt{\frac{\varepsilon}{|M_i+\Delta_q(x,y)|}},
    \end{align*}
    where in the last inequality we also used that $\sqrt{a+b}\le\sqrt{a}+\sqrt{b}$ for $a,b\ge 0$.
\end{proof}

The following lemma shows that, for components $i$, where the regularizer $M_i$ does not guess correctly, we can upper bound the growth of the inverse errors geometrically. 

\begin{lemma}\label{lem:bound1overM}
    Define $B_\tau(x):=\{i\in[n]:|M_i+x|<\tau\}$, then for all $x\in \Reals$
    \[
    \sum_{i\in[n]\setminus B_\tau(x)}\frac{1}{|M_i+x|}\le \frac{2}{\tau}+\frac{2}{\delta}\log\left(1+\frac{n\delta}{\tau}\right).
    \]
\end{lemma}
\begin{proof}
   Define $a_i:=M_i+x$ and observe that $a_i$ is an arithmetic progression with difference $\delta$. Let $B^+:=\{i\in[n]:a_i\ge \tau\}$ and $B^-:=\{i\in[n]:a_i\le-\tau\}$, and observe that $[n]\setminus B_\tau(x)=B^+\cup B^-$, and that $B^+\cap B^-=\emptyset$. Thus
   \[
   \sum_{i\in[n]\setminus B_\tau(x)}\frac{1}{|M_i+x|}= \sum_{i\in B^+}\frac{1}{a_i}+\sum_{i\in B^-}\frac{1}{-a_i}.
   \]

   Then, we can re-index the elements of $\{a_i\}_{i\in B^+}$ to a set $\{a^+_i\}_{i\in [m]}$ (with $m:=|B^+|$), such that $\tau\le a^+_1<a^+_2<\ldots<a_m^+$, and $a^+_k\ge \tau+(k-1)\delta$, thus
   \[
   \sum_{i\in B^+}\frac{1}{a_i}\le \sum_{k=1}^m\frac{1}{a^+_k}\le \sum_{k=1}^m\frac{1}{\tau+(k-1)\delta}\le \sum_{k=0}^{n-1}\frac{1}{\tau+ k\delta}.
   \]
   Similarly, for $i\in B^-$ (now $m=|B^-|$), we can reorder the $a_i$'s into $a_i^-$ in reverse order $-\tau\ge a_1^-\ge a_2^-\ge \ldots\ge a_m^-$ and $a_k^-\le -\tau-(k-1)\delta$ so that
   \[
      \sum_{i\in B^-}\frac{1}{-a_i}\le \sum_{k=1}^m\frac{1}{-a^-_i}\le \sum_{k=0}^{n-1}\frac{1}{\tau+ k\delta}.
   \]
   Finally, we can combine the two to obtain
   \[
   \sum_{i\in[n]\setminus B_\tau(x)}\frac{1}{|M_i+x|}\le2\sum_{k=0}^{n-1}\frac{1}{\tau+ k\delta}\le \frac{2}\tau+2\int_0^{n-1}\frac{1}{\tau+\delta t}dt\le\frac2\tau+\frac{2}{\delta}\log\left(1+\frac{\delta n}{\tau}\right),
   \]
   which is the desired result.
\end{proof}

\begin{lemma}\label{lem:boundB}%
Assume $\delta\le \tau$ and $\tau\le n\delta/8$. Then
\begin{enumerate}[label=(\roman*)]
    \item For all $x\in\Reals$ we have
    \(
    |B_\tau(x)|\le \left\lceil\frac{2\tau}\delta\right\rceil;
    \)
    \item and for all $x\in\Reals$ such that $|x|\le n\delta/4$ it holds that
    \(
    |B_\tau(x)|\ge \frac\tau\delta.
    \)
\end{enumerate}
\end{lemma}

\begin{proof}
    $\{M_i+x\}_{i\in[n]}$ is an arithmetic progression with difference $\delta$. Thus, since $\delta\le2\tau$, every open interval of size $2\tau$ contains at most $\lceil\frac{2\tau}\delta\rceil$ points. This proves that $|B_\tau(x)|\le\lceil\frac{2\tau}\delta\rceil$ for every $x\in \Reals$.

    For the second statement observe that, by the assumption that $\tau\le n\delta/8$ and $|x|\le n\delta/4$, the interval $[-x-\tau,-x+\tau]$ is fully contained in the range of the progression $(M_i)_{i\in[n]}$. Therefore, the number of points of the progression lying in the interval $[-x-\tau,-x+\tau]$ is at least $\lfloor\frac{2\tau}{\delta}\rfloor\ge \frac\tau\delta$.
\end{proof}

Now we can prove the main technical component, which provides a bound on $|H_v(x,y)|$ of the kind $|H_v(x,y)|=O(s_v(x,y)\poly(m/\delta)\log(n))$, thus showing that the magnitude of $|H_v(x,y)|$ is sublinear in $n$. This will be used in showing that the noise term $|\Delta_q(x,y)|$ is small enough so that it falls strictly into the range of the regularizer.

\begin{lemma}\label{lem:boundH2}
For any solution $(x,y)$ to \GDA and any $v\in V$ we have:
    \[
    |H_v(x,y)|\le \frac{16}{\delta}m^2e^{13}s_v(x,y)\log(2n\delta)+2mn\sqrt{\varepsilon}.
    \]
\end{lemma}
\begin{proof}
    Firs notice that our choice of parameters $\delta,n$ satisfy the assumptions of \Cref{lem:boundB} with $\tau=1$.

    If $i\in B_1(\Delta_v(x,y))$, then $\|x_i^v-y_i^v\|_1\le m$, and by \Cref{lem:boundB} item \emph{(i)}, $|B_1(\Delta_v(x,y))|\le 1+2/\delta\le 4/\delta$. Thus, when $i\in B_1(\Delta_v(x,y))$ we get:
    \begin{align}\label{eq:67891}
        \sum_{i\in B_1(\Delta_v(x,y))} \|x_i^v-y_i^v\|_1\le \frac{4 m}\delta
    \end{align}

    On other hand, if $i\notin B_1(\Delta_v(x,y))$, then $|M_i+\Delta_v(x,y)|\ge 1$, and we have:
    \begin{align*}
    \|x_i^v-y_i^v\|_1&=\sum_{j\in[m]}|x_{i,j}^v-y_{i,j}^v|\\
    &\le m^2e^{13}\frac{s_v(x,y)}{|M_i+\Delta_v(x,y)|}+m\sqrt{\varepsilon}\tag{by \Cref{lem:xminusy}}.
    \end{align*}
    Moreover, by \Cref{lem:bound1overM} with $\tau=1$, we can sum over these component to obtain:
    \begin{align}
    \sum_{i\not\in B_1(\Delta_v(x,y))}\|x_i^v-y_i^v\|_1&\le m^2e^{13}s_v(x,y)\sum_{i\not\in B_1(\Delta_v(x,y))}\frac{1}{|M_i+\Delta_v(x,y)|}+mn\sqrt{\varepsilon}\notag\\
    &\le \frac{4}{\delta}m^2e^{13}s_v(x,y)\log(2n\delta)+mn\sqrt{\varepsilon}\label{eq:tmp7890},
    \end{align}
    where in the last inequality we used that $n\delta\ge 1$ and $1=\tau\ge \delta$.%
    Finally, recalling that $H_v(x,y)=\sum_{i\in[n]} \langle F(\frac{x_i^v+y_i^v}{2})-\frac{x_i^v+y_i^v}{2},y_i^v-x_i^v\rangle$ and combining \Cref{eq:67891} and \Cref{eq:tmp7890} we get:
    \begin{align*}
        |H_v(x,y)|&\le \sum_{i\in[n]}\left\|F\left(\frac{x_i^v+y_i^v}{2}\right)-\frac{x_i^v+y_i^v}{2}\right\|_\infty\|x_i^v-y_i^v\|_1\\
        &\le 2\sum_{i\in[n]}\|x_{i}^v-y_{i}^v\|_1\\
        &=2\sum_{i\in B_1(\Delta_v(x,y))}\|x_{i}^v-y_{i}^v\|_1+2\sum_{i\notin B_1(\Delta_v(x,y))}\|x_{i}^v-y_{i}^v\|_1\\
        &\le 8\frac{m}\delta+2\sum_{i\notin B_1(\Delta_v(x,y))}\|x_{i}^v-y_{i}^v\|_1\tag{by \Cref{eq:67891}}\\
        &\le 8\frac{m}\delta+\frac{8}{\delta}m^2e^{13}s_v(x,y)\log(2n\delta)+2mn\sqrt{\varepsilon}\tag{by \Cref{eq:tmp7890}}\\
        &\le \frac{16}{\delta}m^2e^{13}s_v(x,y)\log(2n\delta)+2mn\sqrt{\varepsilon},
    \end{align*}
    concluding the proof.
\end{proof}

\subsection{Small Signal} 
\begin{lemma}\label{lem:smallsignal}
    Consider a solution $(x,y)$ of \GDA and any node $q$ such that $s_q(x,y)=0$, then $b(q)=0$.
\end{lemma}
\begin{proof}
    From \Cref{lem:boundB}, we get that $|B_\delta(\Delta_q(x,y))|\le 2$ and, trivially, also that for all $i\in B_{\delta}(\Delta_q(x,y))$ we have $\|x_{i}^q-y_{i}^q\|^2\le m$.
    On the other hand, if $i\notin B_{\delta}(\Delta_q(x,y))$ then we can invoke \Cref{lem:xminusy}, with $s_q(x,y)=0$, which implies that
    \begin{align*}
        |x_{i,j}^q-y_{i,j}^q|&\le \sqrt{\frac{\varepsilon}{|M_i+\Delta_q(x,y)|}}\le\sqrt\frac\varepsilon\delta,
    \end{align*}
    since for the $i$'s we are considering $|M_i+\Delta_q(x,y)|\ge\delta$. This implies that
    \[
    \|x_{i}^q-y_{i}^q\|^2=\sum_{j\in [m]}(x_{i,j}^q-y_{i,j}^q)^2\le m{\varepsilon}/\delta.
    \]
    Summing the two cases we get 
    \[
    \|x^q-y^q\|^2\le {2m}+ mn{\varepsilon}/\delta\le 3m,
    \]
    since our choice of parameters guarantees that $\varepsilon\le \delta/n$. The proof is concluded by the definition of the decoded solution $b$ and the definition of $\cE_q$
\end{proof}

\subsection{Large Signal}\label{sec:largesign}

\begin{lemma}\label{lem:manyB}
Let $(x,y)$ be a solution to \GDA. Then:
    \[
    |B_1(\Delta_q(x,y)|\ge 1/\delta.
    \]
\end{lemma}
\begin{proof}
    This follows easily from \Cref{lem:boundH2}.
    Indeed from the decomposition of \Cref{lem:decomp}:
    \begin{align*}
    |\Delta_q(x,y)|&\le\sum_{w\in\Out(q)} |H_w(x,y)|\cdot|\phi'_{3m,3m+1}(\|x^q-y^q\|^2)|\cdot|\partial_{\cE_q(x,y)} s_w(x,y)|\\
    &\le me^{14}\left(\frac{16}{\delta}m^2e^{13}\log(2n\delta)+2mn\sqrt{\varepsilon}\right)\le \frac{n\delta}4.
    \end{align*}
    where we used $|\Out(q)|\le m$, $\|\phi_{3m,3m+1}'\|_\infty\le e^2$ and $\|s_v'\|_\infty\le e^{12}$.
    The proof is then concluded by straightforward application of the second item of \Cref{lem:boundB}.
\end{proof}

\begin{lemma}\label{lem:bigsignal}
    Consider any $q\in V$ such that $s_q(x,y)=1$, then, if there is no $i\in[n]$ such that $\tfrac{x_{i}^q+y_i^q}{2}$ is a $\rho$-approximate solution to $\SB$, then $b(q)=1$.
\end{lemma}
\begin{proof}%
    In this proof, for all $i\in[n],j\in[m]$ and $q\in V$, it is convenient to define $\xi^q_i:=\tfrac{x_{i}^q+y_i^q}{2}\in[0,1]^m$ and $\xi^q_{i,j}:=\tfrac{x_{i,j}^q+y_{i,j}^q}{2}\in[0,1]$.
    By \Cref{lem:manyB}, there are at least $1/\delta$ many $i\in B_1(\Delta_q(x,y))$. We now show that for such $i$'s we have $\|x_i^q-y_i^q\|_1\ge \frac{\rho^2}{10 e^{13}}$.
    By contradiction, assume that $\|x_i^q-y_i^q\|_1< \frac{\rho^2}{10 e^{13}}$. \Cref{lem:derivH} let us write for all $j\in[m]$:
    \begin{align}\label{tmp:1357}
        \partial_{x_{i,j}^q} H_q(x,y) &= -(F_j(\xi_i^q)-\xi_{i,j}^q)+R_{i,j,q}(x,y),
    \end{align}
    and,
    \begin{align}\label{tmp:2468}
        \partial_{y_{i,j}^q} H_q(x,y) &= F_j(\xi_i^q)-\xi^q_{i,j}+R_{i,j,q}(x,y),
    \end{align}
    where $|R_{i,j,q}(x,y)|\le \frac{e^{12}}2\|x^q_{i}-y^q_i\|_1\le \frac{\rho^2}{10}$, where the first inequality is by \Cref{lem:derivH}, while the second inequality holds by contradiction' assumption.
    Then consider, any $j\in[m]$, and the optimality conditions of the first player together with \Cref{tmp:1357}:
    \[
    \left(F_j(\xi_i^q)-\xi_{i,j}^q-R_{i,j,q}-2(M_i+\Delta_q(x,y))(x_{i,j}^q-y_{i,j}^q)\right)(z-x_{i,j}^q)\le\varepsilon,
    \]
    for all $z\in[0,1]$ which implies that
    \begin{align*}
    (F_j(\xi_i^q)-\xi_{i,j}^q)(z-x_{i,j}^q)&\le \varepsilon + (R_{i,j,q}+2(M_i+\Delta_q(x,y))(x_{i,j}^q-y_{i,j}^q))(z-x_{i,j}^q)\\
    &\le \varepsilon + \frac{\rho^2}{10}+2|M_i+\Delta_q(x,y)|\cdot|x_{i,j}^q-y_{i,j}^q|\tag{$|z-x_{i,j}^q|\le 1$}\\
    &\le \varepsilon +\frac{\rho^2}{10}+2\frac{\rho^2}{10e^{13}}\le \rho^2\tag{$i\in B_1(\Delta_q(x,y))$}.
    \end{align*}
    which can also be written as
    \begin{align}\label{eq:tmpfirstplayer}
    (F_j(\xi_i^q)-\xi_{i,j}^q)(F_j(\xi_i^q)-x_{i,j}^q)\le \rho^2,
    \end{align}
    after specializing it for $z=F_j(\xi_i^q)\in[0,1]$.

    Similarly, considering the optimality condition of the second player together with \Cref{tmp:2468}, we get for all $z\in[0,1]$:
    \[
    \left(F_j(\xi_i^q)-\xi_{i,j}^q+R_{i,j,q}-2(M_i+\Delta_q(x,y))(x_{i,j}^q-y_{i,j}^q)\right)(z-y_{i,j}^q)\le \varepsilon
    \]
    which implies that
    \begin{align*}
        \left(F_j(\xi_i^q)-\xi_{i,j}^q\right)(z-y_{i,j}^q)&\le\varepsilon+(R_{i,j,q}-2(M_i+\Delta_q(x,y))(x_{i,j}^q-y_{i,j}^q))(y_{i,j}^q-z)\\
        &\le \varepsilon+\frac{\rho^2}{10}+2|x_{i,j}^q-y_{i,j}^q|\le \varepsilon +\frac{\rho^2}{10}+2\frac{\rho^2}{10e^{13}}\le \rho^2,
    \end{align*}
    which is
    \begin{align}\label{eq:tmpsecondplayer}
        (F_j(\xi_i^q)-\xi_{i,j}^q)(F_j(\xi_i^q)-y_{i,j}^q)\le \rho^2,
    \end{align}
    after specializing it for $z=F_j(\xi_i^q)$.

    By averaging \Cref{eq:tmpfirstplayer} with \Cref{eq:tmpsecondplayer} we get:
    \[
    (F_j(\xi_i^q)-\xi_{i,j}^q)^2\le\rho^2\implies |F_j(\xi_i^q)-\xi_{i,j}^q|\le\rho\quad\forall j\in[m].
    \]
    which implies that $\xi_i^q=\frac{x_i^q+y_i^q}{2}$ is a solution to \SB, which is a contradiction. Thus there are at least $1/\delta$ many $i$ such that $\|x_i^q-y_i^q\|_1\ge \frac{\rho^2}{10e^{13}}$, and
    \begin{align*}    
    \|x^q-y^q\|^2&\ge \sum_{i\in B_1(\Delta_q(x,y))}\|x_i^q-y_i^q\|^2\\
    &\ge \frac{1}{m}\sum_{i\in B_1(\Delta_q(x,y))}\|x_i^q-y_i^q\|_1^2\\
    &\ge \frac{\rho^4}{100 e^{26}\delta m}=4m\ge 3m+1,
    \end{align*}
    concluding the proof, since by definition $b(q)=\cE_q(\|x^q-y^q\|^2)=1.$%
\end{proof}

\subsection{Proof of \Cref{lem:dicotomy}}
\dichotomy*
\begin{proof}
By \Cref{lem:bigsignal} we have that either there is a $(i,q)\in[n]\times V$ such that $(x_i^q+y_i^q)/2$ is a solution to \SB, or that for all $q\in V$ such that $s_q(x,y)=1$ then $b(q)=1$. On the other hand, \Cref{lem:smallsignal} always guarantees that $s_q(x,y)=0$ implies that $b(q)=0$. Thus, we can assume that $s_v(x,y)=1\implies b(v)=1$ and $s_v(x,y)=0\implies b(v)=0$ for all $v\in V$.
Now, we show that in this case, the assignment $b$ satisfies all the gates:
\begin{itemize}[leftmargin=0.5cm]
\item  $(u,v,w)\in\cG_{\NOR}$. Assume that $\cE_u(x,y)=\cE_v(x,y)=b(u)=b(v)=0$, then $s_w(x,y)=g(\cE_u(x,y)+\cE_v(x,y))=g(0)=1$. By assumption, this implies that also $b(w)=1$ and thus the gate is satisfied.
On the other hand, if either $b(u)=1$ or $b(v)=1$, then $\cE_u(x,y)+\cE_v(x,y)\ge 1$ and thus $s_w(x,y)=0$. By assumption, this implies that $b(w)=0$, and thus \NOR gates are satisfied.
\item $(u,v,w)\in\cG_{\PURIFY}$.
If $b(u)=0$ then $\cE_u(x,y)=0$ and thus $s_v(x,y)=s_w(x,y)=\ell(\cE_u(x,y)\pm1/4)=0$, which by assumption implies that $b(v)=b(w)=0$.
Similarly, if $b(u)=1$ then $\cE_u(x,y)=1$ and $s_v(x,y)=s_w(x,y)=\ell(\cE_u(x,y)\pm1/4)=1$, which by assumption implies that $b(v)=b(w)=1$. On the other hand, if $b(u)=\bot$, then $\cE_u(x,y)\in(0,1)$ but either $\cE_u(x,y)-1/4\le 5/12$ or $\cE_u(x,y)+1/4\ge 7/12$, which implies that either $s_v(x,y)=0$ or $s_w(x,y)=1$, thus satisfying the gate's constraints.

\item $(u_1,\ldots, u_N,v)\in\cG_{\ORACLE}$.
Assume that $(b(u_1),\ldots, b(u_N))\in\{0,1\}^N$. Then $\cE_{u_i}(x,y)=b(u_i)$ for all $i\in[N]$, and $s_v(x,y)=\cH((b(u_i))_{i\in[N]},\cL)$. By \Cref{lem:inteprolation}, we get that $s_v(x,y)=\cL(b(u_1),\ldots, b(u_N))\in\{0,1\}$, which trivially shows that the gate is satisfied. By assumption we have that $b(v)=\cL(b(u_1),\ldots, b(u_N))$ as desired.
\end{itemize}

This proves that $b$ is a valid assignment to \OPC. 
\end{proof}

\printbibliography
\end{document}